%% file: main3.tex
\documentclass[12pt,oneside]{amsbook}
\usepackage[T1]{fontenc}
\usepackage[latin9]{inputenc}
\usepackage{geometry}
\usepackage{fancyhdr}
\usepackage{xcolor}
\usepackage{pdfcolmk}
\usepackage{amsthm}
\usepackage{amstext}
\usepackage{amssymb}
\usepackage{graphicx}
\usepackage{setspace}
\usepackage{esint}
\PassOptionsToPackage{normalem}{ulem}
\usepackage{ulem}
\usepackage[unicode=true]
 {hyperref}

\makeatletter

\newcommand{\noun}[1]{\textsc{#1}}
\providecolor{lyxadded}{rgb}{0,0,1}
\providecolor{lyxdeleted}{rgb}{1,0,0}
\newcommand{\lyxadded}[3]{{\texorpdfstring{\color{lyxadded}{}}{}#3}}
\newcommand{\lyxdeleted}[3]{{\texorpdfstring{\color{lyxdeleted}\sout{#3}}{}}}

\numberwithin{section}{chapter}
\theoremstyle{plain}
\newtheorem{thm}{\protect\theoremname}
  \theoremstyle{definition}
  \newtheorem{defn}[thm]{\protect\definitionname}
  \theoremstyle{remark}
  \newtheorem{rem}[thm]{\protect\remarkname}
  \theoremstyle{plain}
  \newtheorem{lem}[thm]{\protect\lemmaname}
  \theoremstyle{plain}
  \newtheorem{prop}[thm]{\protect\propositionname}
\newenvironment{lyxlist}[1]
{\begin{list}{}
{\settowidth{\labelwidth}{#1}
 \setlength{\leftmargin}{\labelwidth}
 \addtolength{\leftmargin}{\labelsep}
 }}
{\end{list}}

\numberwithin{equation}{chapter}

\usepackage{amsmath}

\usepackage{microtype}

\RequirePackage{booktabs} 

\RequirePackage{textcase} 

\RequirePackage{soul} 

\sodef\allcapsspacing{\upshape}{0.15em}{0.65em}{0.6em}

\sodef\lowsmallcapsspacing{\scshape}{0.075em}{0.5em}{0.6em}

\usepackage{tabularx}

\usepackage{geometry}

\usepackage{mathpazo}

\usepackage{microtype}

\makeatother

  \providecommand{\definitionname}{Definition}
  \providecommand{\lemmaname}{Lemma}
  \providecommand{\propositionname}{Proposition}
  \providecommand{\remarkname}{Remark}
\providecommand{\theoremname}{Theorem}

\begin{document}
\frontmatter
\pagestyle{empty}

\begin{center}
\textsc{\LARGE Application of Geometric measure}\textsc{\large }\textsc{\LARGE{}
Theory in Continuum Mechanics:}
\par\end{center}{\LARGE \par}

\begin{center}
\textsc{\large The Configuration Space, Principle of Virtual Power
and Cauchy's Stress Theory for Rough Bodies}
\par\end{center}{\large \par}

\vspace{5mm}

\begin{center}
Thesis submitted in partial fulfillment \\
of the requirements for the degree of\\
``DOCTOR OF PHILOSOPHY''
\par\end{center}

\begin{center}
\includegraphics[scale=0.3]{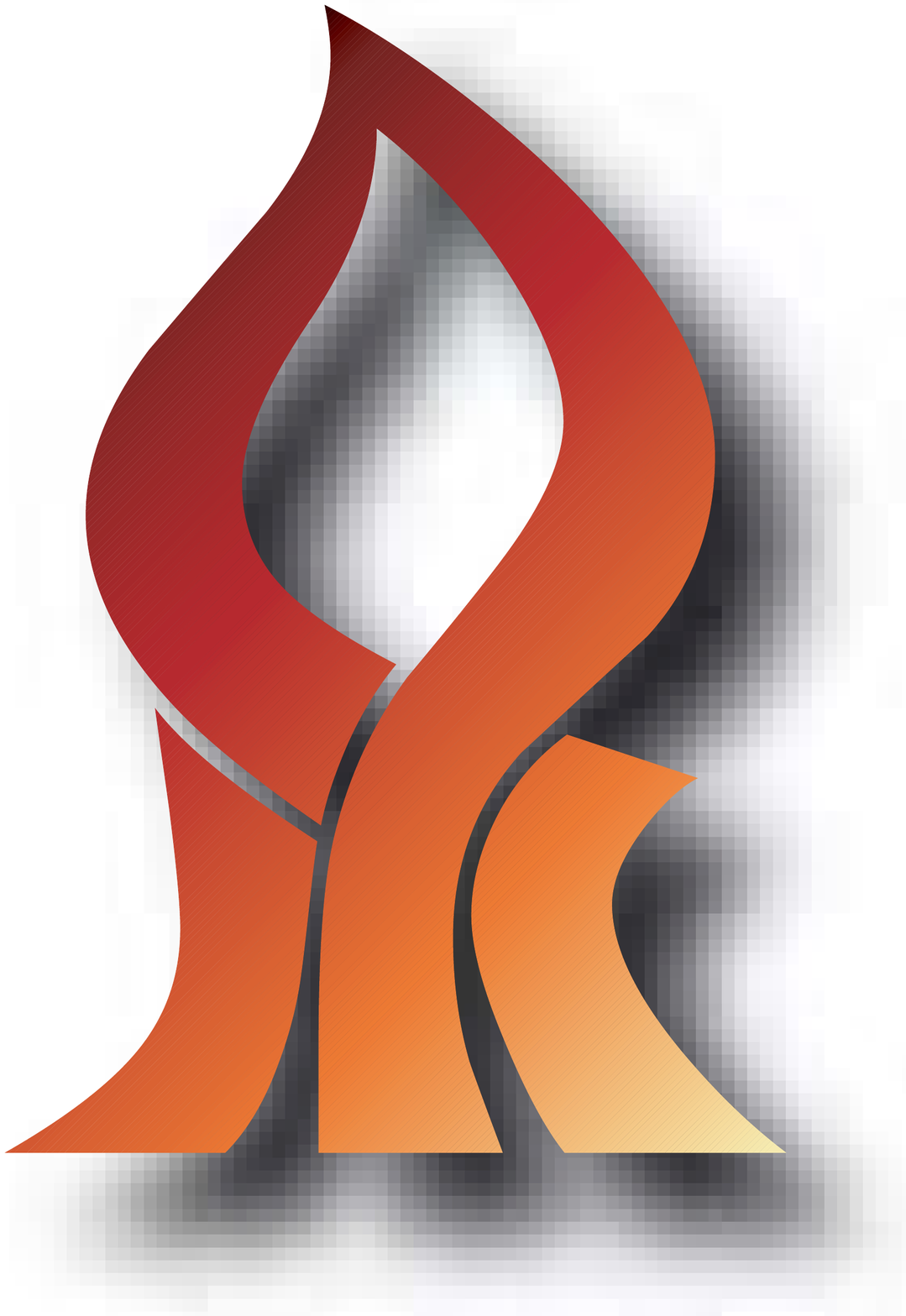}
\par\end{center}

\begin{center}
\noun{\large By }
\par\end{center}{\large \par}

\begin{center}
{\huge \vspace*{5mm}}\textsc{\LARGE Lior Falach}
\par\end{center}{\LARGE \par}

\begin{center}
{\huge \thispagestyle{empty}\vspace*{5mm}}\textsc{\noun{\large Advisor:
}}\textsc{\large Prof.~Reuven Segev}
\par\end{center}{\large \par}

\vspace{0.5cm}

\begin{center}
Submitted to the Senate of Ben-Gurion University of the Negev
\par\end{center}

\vspace{0.5cm}

\begin{center}
\emph{\large May 2013}
\par\end{center}{\large \par}

\begin{center}
\noun{\large Beer-Sheva}
\par\end{center}{\large \par}

\pagebreak{}

\begin{center}
{\large }\textsc{\LARGE Application of Geometric measure Theory in
Continuum Mechanics:}{\LARGE{} }
\par\end{center}{\LARGE \par}

\begin{center}
\textsc{\large The Configuration Space, Principle of Virtual Power
and Cauchy's Stress Theory for Rough Bodies}
\par\end{center}{\large \par}

\vspace{1cm}

\begin{center}
Thesis submitted in partial fulfillment \\
of the requirements for the degree of\\
``DOCTOR OF PHILOSOPHY''
\par\end{center}

\begin{center}
\vspace{1cm}
\noun{\large By }
\par\end{center}{\large \par}

\begin{center}
\vspace{0.5cm}
{\huge Lior Falach}
\par\end{center}{\huge \par}

\vspace{1cm}

\begin{center}
Submitted to the Senate of Ben-Gurion University\\
of the Negev
\par\end{center}

\vspace{1cm}

Approved by the advisor

\vspace{1cm}

Prof. Reuevn Segev:\hfill{} \uline{~~~~~~~~~~~~~~~~~~~~~~~~~~~~~~~~~}
Date: \uline{~~~~~~~~~~~~~~~~~~~~~~~~~~~~~~~~~}

\vspace{1cm}

Approved by the Dean of the Kreitman School of Advanced Graduate Studies 

\vspace{1cm}

Prof. Michal Shapira:\hfill{} \uline{~~~~~~~~~~~~~~~~~~~~~~~~~~~~~~~~~}
Date: \uline{~~~~~~~~~~~~~~~~~~~~~~~~~~~~~~~~~}

\vfill{}

\begin{center}
\emph{\large May 2013}
\par\end{center}{\large \par}

\begin{center}
\noun{\large Beer-Sheva}
\par\end{center}{\large \par}

\pagebreak{}

This work was carried out under the supervision of 

Prof. Reuevn Segev

In the Department of Mechanical Engineering

Faculty of Engineering Science 

\pagebreak{}

\section*{Research-Student's Affidavit when Submitting the Doctoral Thesis
for Judgment}

I \uline{Lior Falach}, whose signature appears below, hereby declare
that\\
 (Please mark the appropriate statements):

\vspace{1cm}

\_\_\_ I have written this Thesis by myself, except for the help and
guidance offered by my Thesis Advisors.

\vspace{1cm}

\_\_\_ The scientific materials included in this Thesis are products
of my own research, culled from the period during which I was a research
student.

\vspace{1cm}

\_\_\_ This Thesis incorporates research materials produced in cooperation
with others, excluding the technical help commonly received during
experimental work. Therefore, I am attaching another affidavit stating
the contributions made by myself and the other participants in this
research, which has been approved by them and submitted with their
approval.

\vspace{3cm}

Date: \_\_\_\_\_\_\_\_\_\_\_\_\_\_\_\_\_ 

\vspace{1cm}

Student's name: \_\_\_\_\_\_\_\_\_\_\_\_\_\_\_\_ 

\vspace{1cm}

Signature:\_\_\_\_\_\_\_\_\_\_\_\_\_\_

\pagebreak{}

\section*{Acknowledgment}

First and foremost I wish to express my heartfelt gratitude to Professor
Reuven Segev. Above all academic merits and scientific pursuit of
excellence I found in Reuven a mentor and a dear friend. For his guidance,
inspiration, support, encouragement and endless patient I will forever
be grateful and believe that I could not have asked for a better role
model.

Acknowledgment are also due to Prof. Gal deBotton and Prof. Michael
Mond for the guidance over the past years. My friends Alon and Kobi,
thank you for your attentive ears and patiantce.

To my dear wife Tal, (A.K.A Talula) who has been by my side and taken
every single step of this Ph.D., I would like to thanks for her understanding
and loving support during the past few years. Her support and encouragement
was, in the end, what made this dissertation possible. My dear son
Itamar, (A.K.A little Federer) thank you for the great joy you brought
into our life. 

\vfill{}

For all those uncountably many who aided in the completion of this
thesis, thank you for allowing me to bask in your light.
\begin{quote}
``We are all meant to shine, as children do. We were born to make
manifest the glory of God that is within us. It's not just in some
of us; it's in everyone. And as we let our own light shine, we unconsciously
give other people permission to do the same.'' {[}Marianne Williamson{]}
\end{quote}
\pagebreak{}

\bigskip{}

\begin{center}
\textit{This work is dedicated to my parents,}
\par\end{center}

\begin{center}
\textit{Ester and Eliyaho Falach,}
\par\end{center}

\begin{center}
\textit{My heroes.}
\par\end{center}

\vfill{}

\tableofcontents{}

\section*{Abstract}

It is a generally agreed upon notion that the elements of geometric
measure theory should play a central role in the mathematical formulation
of continuum mechanics. Homological integration theory, a branch of
geometric measure theory and differential geometry, is concerned with
the generalization of the notion of integration on manifolds. The
generalized integral is defined by a current, and within the class
of currents special attention is given to integrals over polyhedral
chains, normal currents and flat chains. Each of the aforementioned
currents may be viewed as an integration over a domain in space of
varying degree of regularity, or, as implied by our present objective,
irregularity. By applying the tools of homological integration theory,
generalized bodies will be introduced. The applicability of homological
integration theory to continuum mechanics has been noted in several
central works as for example:
\begin{quotation}
``... This result is of independent mathematical interest and is
intimately related to the flat form and cochains of Whitney''\cite{Ziemer1983},\end{quotation}
\begin{quote}
``... We note that our development has points of contact with some
ideas of geometric measure theory, the theory of distributions, and
earlier developments of the mathematical foundations of mechanics''\cite{Antman1979}.
\end{quote}
This thesis further explores the applicability of homological integration
theory for the mathematical formulation of continuum mechanics. The
proposed framework is shown to enable the inclusion of a generalized
class of bodies such that a corresponding stress theory is properly
formulated and a generalized principle of virtual power is presented.

In the setting of an $n$-dimensional Euclidean space, an admissible
body is initially viewed a normal $n$-current induced by a set of
finite perimeter. Bodies viewed as normal $n$-currents serve as our
elementary building blocks which are used in the construction of a
generalized Cauchy's flux. The configuration space of bodies in the
physical space is assumed to be comprises Lipschitz embeddings, which
are shown to form an open subset in the space of locally Lipschitz
maps endowed with the strong Lipschitz topology. Thus, virtual velocity
fields are naturally viewed as locally Lipschitz maps. A field over
a body is represented by the multiplication of a sharp function and
a normal current. A density transport theorem is developed which is
shown to be analogous to Reynolds' transport theorem for an implicit
time dependent property.

A generalized Cauchy flux is defined as a real valued function on
the Cartesian product of $(n-1)$-currents representing material surfaces
and locally Lipschitz mappings representing virtual velocities. The
duality between restricted velocity fields and Cauchy fluxes is studied
and a generalized version of Cauchy's postulate implies that a Cauchy
flux may be uniquely extended to an $n$-tuple of flat $(n-1)$-cochains.
Thus, the class of admissible bodies is extended to include flat $n$-chains,
which may be viewed as currents induced by Lebesgue integrable sets.
We note that no restriction is imposed on the measure theoretic boundary
of the generalized body, yet, the flux over the boundary of this generalized
body is well defined as the boundary of a generalized body is viewed
as a flat $(n-1)$-chain. A general subset of the boundary may not
be a flat $(n-1)$-chain. A generalized material surface is formally
introduced as a trace, defined as the intersection of the boundary
with a set of finite perimeter. A trace is shown to be a flat $(n-1)$-chain
thus, the flux across such generalized material surfaces may be calculated.
Wolfe's representation theorem for flat cochains enables the identification
of stress as an $n$-tuple of flat $(n-1)$-forms providing an integral
representation the Cauchy flux.\lyxadded{Reuven Segev}{Wed May 08 21:57:49 2013}{}

\mainmatter
\pagestyle{fancy}

\include{my_macros}

\include{Introduction}

\include{gmt}

\include{Sets_of_finite_perimeter}

\include{Lipschitz}

\include{sharp}

\include{configuration}

\include{Transport}

\include{Cauchy_flux}

\include{generalized_bodies}

\include{Strain}

\include{Stress}

\chapter*{Concluding remarks and further points of research}

This thesis demonstrates again that the fundamental notions of continuum
mechanics may be generalized by applying the tools of geometric measure
theory. Identifying\lyxdeleted{lior}{Thu May 09 13:35:55 2013}{ } bodies
as currents in $\reals^{n}$ led to the definition of Lipschitz embedding
configurations and locally Lipschitz virtual velocities. A generalized
stress theory was presented in which the stress was identified with
an $n$-tuple of flat $(n-1)$-forms. This generalized stress theory
enables the inclusion of $n$-rectifiable sets into the class of admissible
bodies. The class of generalized bodies, viewed as flat $n$-chains,
serve as an extension to the class of sets of finite perimeter representing
taken as the class of admissible bodies. The inclusion of flat $n$-chains
in the class of admissible bodies implies minimal restrictions on
the boundary of bodies. Thus, sets of fractal boundary, rough bodies,
were shown to be admissible bodies. In addition, a density transport
theorem was formulated within the proposed framework and it was shown
to be analogous to Reynolds' transport theorem. 

Further research is suggested in order to investigate further applications
to the theory to the generalization of some fundamental notions such
as:
\begin{itemize}
\item The mechanics of $r$-dimensional bodies for $r<n$ (analogous to
the theory of plates and shells) may be introduced by the examination
of bodies represented by $r$-currents.
\item A generalized Reynolds' transport theorem which 

\begin{itemize}
\item applies to rough bodies,
\item includes explicit time dependent properties by representing the property
by a time dependent cochain.
\end{itemize}

\end{itemize}
In addition, an extension of the theory to the general setting of
differentiable manifolds devoid of any metric structure or parallelism
structure would be an interesting program to pursue.

\backmatter

\bibliographystyle{alpha}

\end{document}

%% file: my_macros.tex
\global\long\def\body{\mathcal{B}}
\global\long\def\part{\mathcal{P}}
\global\long\def\surface{\mathcal{S}}

\global\long\def\stress{X}
\global\long\def\strain{\chi}
\global\long\def\strech{\varepsilon}

\global\long\def\cbnd{d}
\global\long\def\cochain{X}

\global\long\def\force{g}
\global\long\def\bforce{f}
\global\long\def\tforce{t}
\global\long\def\cflux{\Psi}

\global\long\def\bsys{\Psi}
\global\long\def\ssys{\Phi}
\global\long\def\fsys{\hat{\Psi}}

\global\long\def\Lip{\mathfrak{L}}
\global\long\def\SS#1{\mathcal{L}_{\mathfrak{s}}\left(#1\right)}
\global\long\def\Limb{\Lip_{\mathrm{Em}}}
\global\long\def\Lmap{\mathcal{F}}

\global\long\def\Smap{\phi}
\global\long\def\Nlin#1{\|#1\|_{\Lip}}
\global\long\def\gnorm{\odot}

\global\long\def\dist#1{\mathrm{dist}_{#1}}

\global\long\def\con{V}

\global\long\def\reals{\mathbb{R}}
\global\long\def\too{\longrightarrow}
\global\long\def\Vs{V}
\global\long\def\we{\wedge}

\global\long\def\Vr{\Vs_{r}}
\global\long\def\Vrd{\Vs^{r}}

\global\long\def\spt{\mathrm{spt}}

\global\long\def\imag#1{\mathrm{image}\left(#1\right)}

\global\long\def\emb{\varphi}

\global\long\def\emr{g}

\global\long\def\refi{\mathfrak{S}}

\global\long\def\eqr{\overset{\body}{\backsim}}
 \global\long\def\eqc#1{\left[#1\right]}
\global\long\def\qLip{\Lip(\oset,\pspace)_{\body}}
\global\long\def\qnorm#1{\left\Vert #1\right\Vert _{\qLip}}
\global\long\def\qpro{\pi_{\body}}
\global\long\def\ic{i}

\global\long\def\rconf{\varphi}
\global\long\def\rLip{f}
\global\long\def\rvirv{v}
\global\long\def\rmvirv{u}

\global\long\def\eqrconf{\overset{\conf\left(\body\right)}{\backsim}}

\global\long\def\inter#1{\mathrm{int}\left(#1\right)}

\global\long\def\sform#1#2{D^{#1}\left(#2\right)}

\global\long\def\currents#1#2{D_{#1}\left(#2\right)}

\global\long\def\Fmass#1{M\left(#1\right)}

\global\long\def\Fnormal#1{N\left(#1\right)}

\global\long\def\Fcomass#1{\left\Vert #1\right\Vert _{0}}

\global\long\def\Nnorm#1{\left\Vert #1\right\Vert ^{N}}

\global\long\def\var#1{\mathrm{Var}\left(#1\right)}

\global\long\def\ball#1#2{B\left(#1,#2\right)}

\global\long\def\mtb{\Gamma}

\global\long\def\dns#1#2{d\left(#1,#2\right)}

\newcommand{\note}[1]{\textcolor{blue}{**[#1]**}}

\global\long\def\compact{K}

\global\long\def\norm#1#2{\left\Vert #1\right\Vert _{#2}}

\global\long\def\colb#1#2#3{B^{\Lip}(#1,#2,#3)}

\global\long\def\cell{\sigma}
\global\long\def\simp{\sigma}
\global\long\def\oset{U}
\global\long\def\pspace{S}

\global\long\def\chain{\mathcal{A}}
\global\long\def\bnd{\partial}

\global\long\def\mass#1{|#1|}
\global\long\def\fmass#1{|#1|^{\flat}}
\global\long\def\fmassr#1{|#1|_{\oset}^{\flat}}
\global\long\def\smass#1{|#1|^{\sharp}}

\global\long\def\rest{\raisebox{0.4pt}{\,\mbox{\ensuremath{\llcorner}}\,}}
\global\long\def\irest{\raisebox{0.4pt}{\mbox{\,\ensuremath{\lrcorner}\,}}}

\global\long\def\ess{\mathrm{ess}}

\global\long\def\virv{v}
\global\long\def\mvirv{\xi}

\global\long\def\conf{\kappa}
\global\long\def\confs{\mathcal{Q}}
\global\long\def\gconf{\kappa}

\global\long\def\virvs{W_{\conf}}
\global\long\def\mvirvs{W}

\global\long\def\motion{M}

\global\long\def\vMmass#1{\left\Vert #1\right\Vert }

\global\long\def\gpart{\mathring{\part}}
\global\long\def\gsurface{\mathring{\surface}}

\global\long\def\halfs{H}

\global\long\def\gunivb{\mathring{\Omega}_{\body}}
\global\long\def\gunivs{\bnd\mathring{\Omega}_{\body}}

\global\long\def\D{\mathfrak{D}}

\global\long\def\form{\phi}

\global\long\def\Fmass#1{M\left(#1\right)}

\global\long\def\measure#1{\mu_{#1}}

\global\long\def\Fnormal#1{N\left(#1\right)}

\global\long\def\Fflat#1{F\left(#1\right)}

\global\long\def\lusb{L}

\global\long\def\wcbd{\widetilde{\cbnd}}

\global\long\def\fform#1{D_{#1}}

\global\long\def\strech{\varepsilon}

\global\long\def\ivirv{\chi}

\global\long\def\stress{\tau}

%% file: Introduction.tex
\chapter{Introduction\label{sec:Introduction}}

This thesis presents a framework for the formulation of some fundamental
notions of continuum mechanics. Specifically, using elements from
geometric measure and integration theory, we consider, within the
geometric setting of $\reals^{n}$, the class of admissible bodies,
configurations of bodies in space, the configuration space, virtual
velocities, Reynold's transport theorem and Cauchy's stress theory.

Cauchy's stress theorem is one of the central results in continuum
mechanics. It asserts the existence of the stress tensor which determines
the traction fields on the boundaries of the various bodies. As the
traditional proof relies on locality and regularity assumptions, from
both the validity and the applicability aspects, stress theory is
closely associated with the proper choice of the class of bodies.
Furthermore, an appropriate class of bodies should allow the formulation
of the Gauss-Green theorem or a generalization thereof. 

In light of these observations, formulations of the fundamentals of
continuum mechanics have considered, since the middle of the 20th
century, the appropriate choice of the class of bodies. In \cite{Noll1959},
Noll sets an axiomatic scheme for continuum mechanics in which a rigorous
mathematical framework for the concepts of bodies, kinematics, forces
and dynamical processes is presented. A body is defined as a compact,
differentiable three-dimensional manifold with piecewise smooth boundary,
the manifold is assumed to be covered by a single chart and is endowed
with a measure space structure. The configurations of the body in
space provide charts on the body manifold and a part of the body is
defined as a compact subset of the body with piecewise smooth boundary.
The existence of the stress and Cauchy's original postulate on the
dependence of the traction on the exterior normal is shown to follow
from the additivity assumption on the system of forces and the principle
of linear momentum. In \cite[p.~466]{Truesdell1960} Truesdell and
Toupin ignore the formal issue of admissible bodies and tacitly assume
smoothness wherever necessary. Later on, in \cite[p.~4]{Truesdell1966},
Truesdell defines a body as a differential manifold endowed with a
structure of a $\sigma$-finite measure space and the $\sigma$-ring
of subsets are viewed as parts of the body.. The common ground for
these early works is in the assumption that bodies in continuum mechanics
should have a smooth structure so that the classical versions of the
notions of mathematical analysis apply. 

The \textit{material universe}, a formal structure for the class of
admissible bodies, was presented by W. Noll \cite{Noll1962,Noll1973}.
The material universe is assumed to be a partially ordered collection
of sets. The collection is further furnished with the operations $\wedge$
(meet) and $\vee$ (join) (generalizing the operations of intersection
and union operations on sets) such that the material universe has
the structure of a Boolean algebra.

In \cite{Gurtin1967}, M.Gurtin and W. Williams present an axiomatic
formulation of continuum thermodynamics in which a body is viewed
as a standard region, \textit{i.e.,} the closure of a bounded open
set in a three-dimensional  Euclidean space. The body's boundary
is composed of the union of a closed set of zero area measure and
a countable collection of two-dimensional manifolds of class $C^{1}$.
The collection of subbodies defined in \cite{Gurtin1967} has less
structure than Noll's material universe and it is selected such that
the collection of subbodies will enable a proof of the existence of
intrinsic thermodynamical quantities such as the radiation density,
heat flux vector and the internal entropy density.

M. Gurtin and L. Martins introduced in \cite{Gurtin1975} the notion
of a Cauchy flux in order to represent the collection of total forces
applied to the collection of plane  surface elements. A Cauchy flux
is defined as an additive, area bounded set function acting on the
collection of compatible surface elements of the body, and a weakly
balanced Cauchy flux is defined as a volume bounded Cauchy flux. It
is shown that for each plane surface element the density of the Cauchy
flux exists almost everywhere with respect to the Hausdorff area measure.
The weak balance postulate is shown to be a necessary condition for
the linearity of the dependence of the density on the normal as well
as for the formulation of a classical balance law for the Cauchy flux. 

It seems that \cite{Banfi1979} and \cite{Ziemer1983} were the first
to propose that the class admissible bodies in continuum physics should
consist of sets of finite perimeter. In Ziemer's work, admissible
bodies are defined as sets of finite perimeter and a weakly balanced
Cauchy flux is shown to be represented by a measurable vector field.
The works \cite{M.E.Gurtin1986,Noll1988}, which followed, further
extended these studies. In \cite{M.E.Gurtin1986}, the class of admissible
bodies is defined as the class of normalized sets of finite perimeter
while in \cite{Noll1988}, admissible bodies are defined as fit regions
which are bounded regularly open sets of finite perimeter and of negligible
boundary. These postulates enabled the authors to apply a version
of the Green-Gauss theorem and consider sets that do not necessarily
have smooth boundaries as bodies in continuum mechanics for which
balance laws may be written. 

In \cite{Silhavy1985,Silhavy1991}, Silhavy considered bodies as sets
of finite perimeter in a bounded open region of $\reals^{n}$. The
author employs a weak approach in the formulation of Cauchy's flux
theorem. Silhavy's approach gives rise to a Borel set $N_{0}$, of
Lebesgue measure zero and a flux vector field $q$, such that the
action of the Cauchy flux is represented the by $q$ for any surface
whose intersections with $N_{0}$ has Hausdorff area measure zero.
The analysis presented in Silhavy's work allows for singularities
in the flux vector field and presents for the first time the concept
of almost every surface. In \cite{Silhavy1991}, formal definitions
of the concepts of almost every body and almost every surface are
given and the choice of the class of admissible bodies is shown to
be intimately related with the class of representing flux vector fields.
The notions of almost every body and almost every surface are further
examined in \cite{Degiovanni1999} and it is shown that the Cauchy
flux is determined by its action on a collection of rectangular planar
surfaces with edges parallel to the axes of $\reals^{n}$. A similar
extension of the Cauchy interaction is presented in \cite{Marzocchi2003}.
In the above works it is shown that a weakly balanced Cauchy flux
is represented by a divergence-measure vector field. 

Vector fields of bounded variations are vector fields whose components
are Radon measure and whose all partial derivatives, taken in the
distributional sense, are Radon measures. Divergence-measure vector
fields are viewed as a generalization of bounded variation vector
fields such that one requires each component to be represented by
a Radon measure and the divergence, taken in the distributional sense,
is represented by a Radon measure (rather than all partial derivatives).
The Gauss-Green theorem has been established for sets of finite perimeter
and functions of bounded variations by Federer \cite[Section 4.5]{Federer1969}.
Similar results were obtained for divergence-measure vector field
and sets of Lipschitz deformable boundaries in \cite{Chen1999,Chen2001}.
The generalization of these results for sets of finite perimeter were
obtained in \cite{Chen2005,Chen2009,Silhavy2005}, and a further generalization
of the theories for sets of fractal boundaries is given in \cite{Silhavy2009}.
In \cite{Silhavy2005}, Cauchy's flux theory is developed for sets
of finite perimeter where it is shown that a real valued Cauchy flux
is represented by a divergence measure vector field. The development
is extended in \cite{Silhavy2006} where fluxes over parts fractal
boundaries are investigated and the notion of normal trace for sets
of fractal boundaries is introduced. \textit{Rough bodies,} introduced
by Silhavy \cite{Silhavy2006}, are sets whose measure theoretic boundaries
are fractals in the sense that the outer normal is not defined almost
everywhere with respect to the $(n-1)$-Hausdorff measure.

In \cite{Segev1986}, a weak formulation of $p$-grade continuum mechanics,
for any integer $p\ge1$, is presented in the setting of differential
manifolds. Configurations are viewed as $C^{p}$-embeddings of the
body manifold in the physical space and forces are viewed as elements
of the cotangent bundle to the infinite dimensional configuration
manifold of mappings. Forces are shown to be represented by measures
on the $p$-th jet bundle. Such a measure serves as a generalization
of the $p$-th order stress. The representation of forces by stress
measures enables a natural restriction of forces to subbodies. The
consistency conditions for a such a system of $p$-th order forces
are examined in \cite{Segev1991}.

The term \textit{fractal} was coined in 1975 by Mandelbrot to indicate
a highly irregular geometric object (see \cite{Mandelbrot1983}).
Mandelbrot's seminal work was the beginning of a very large body of
research concerning the fractal properties of various physical phenomena.
A variety of approaches have been suggested for the adaptation of
fractal objects to branches of mechanics, \textit{e.g.}, \cite{Tarsov2005,Tarsov2005a,Tarsov2005b,Tarsov2005c,Wnuk2008,Wnuk2009,Epstein2006,Ostoja-Starzewski2009}.

In \cite{Rodnay2002,Rodnay2003}, Cauchy's flux theory is formulated
using Whitney's geometric integration theory \cite{Whitney1957} and
new developments by Harrison \cite{Harrison93,Harrison98a,Harrison98b,Harrison99}.
Bodies are viewed as $r$-dimensional domains of integration in an
$n$-dimensional Euclidean space with $r\leq n$. A body is identified
as an $r$-chain, the limit of a sequence of polyhedral chains with
respect to a norm which is induced by Cauchy's postulates. Three types
of chains are examined: flat, sharp and natural chains, such that
\[
\text{polyhedral}\subset\text{flat chains}\subset\text{sharp chains}\subset\text{natural chains.}
\]
Flat $(n-1)$-chains may represent the fractal boundaries of bodies
and sharp chains are shown to represent even less regular $(n-1)$-dimensional
objects. Fluxes of a given extensive property are postulated to be
$(n-1)$-cochains, \textit{i.e.,} elements of the dual to the Banach
space of $(n-1)$-chains. By the duality structure of Whitney's theory,
as one allows for less regular domains of integration (chains), the
resulting fluxes (cochains) become more regular, automatically. 

The present work, describes a framework where the mechanics of bodies
with fractal boundaries may be studied. Unlike \cite{Rodnay2003},
in which Whitney's geometric integration theory is applied, in this
work the point of view of geometric measure theory as in \cite{Federer1969}
is mainly adopted. Geometric measure theory can be described best
as a generalization of differential geometry by means of measure theory
with the purpose of dealing with non-smooth maps and surfaces. Geometric
measure theory has a mechanical-like origin in Plateau's problem which
considers surfaces of minimal area having given boundaries as models
of soap bubbles. For the relation between Geometric measure theory
and the Plateau problem, see \cite{Almgren1966}. Early contributions
to geometric measure theory may be attributed to such mathematicians
as L.C.~Young and E.~De~Giorgi (see \cite{Giorgi2006}). The theory
has taken a formal structure in Federer \& Fleming's seminal paper
\cite{Federer1960}. As a reference on Geometric measure theory one
may use Federer's \cite{Federer1969} or \cite{Krantz2008,Giaquinta1998,Morgan2008,Simon1984,Lin2002}.

The universal body is modeled as an open subset of $\reals^{n}$ and
bodies are modeled as flat chains. In addition to the properties of
the class of admissible bodies, special attention is given to the
study of the kinematics of such bodies in space. The appropriate class
of admissible configurations appears to be the set of Lipschitz embeddings.
This class enjoys two significant properties. Firstly, the set of
Lipschitz embeddings of the universal body into space is an open subset
of the locally convex topological vector space of all Lipschitz mappings
of the universal body into space equipped with the Whitney, or strong,
topology. In addition, for Lipschitz mappings there is a well defined
pushforward action on flat chains, such as those representing bodies.
Therefore, the images of bodies under the pushforward action induced
by a Lipschitz embedding preserve their structure and relevant properties
(\emph{e.g.}, the availability of a generalized Stokes theorem).

Adopting the point of view that virtual velocities are elements of
the tangent bundle of the configuration manifold, as the configuration
space is open in the space of Lipschitz mappings, virtual velocities
may be identified with Lipschitz mappings of the universal body into
space. Considering force and stress theory, it is noted that forces
which are required only to be continuous linear functionals relative
to the Lipschitz topology, as would be the analogue of \cite{Segev1986},
seem to be too irregular for the setting adopted here. In order to
constitute a consistent force system which is represented by an integrable
stress fields, balance and weak balance are postulated. It is shown
further that balance and weak balance are equivalent together to continuity
relative to the flat norm of chains.

The thesis is constructed as follows. Chapters \ref{sec:Federer}--\ref{sec:Sharp_functions}
contain a short outline of the various notions of geometric measure
theory which are used in this work. Chapter \ref{sec:Federer} reviews
the notion of differential forms, currents, flat chains and cochains.
Chapter \ref{sec:Sets_of_finite_perimeter} presents sets of finite
perimeter as well as the corresponding definitions for bodies and
material surfaces as currents. In Chapter \ref{sec:On-Lipschitz-mappings}
we discuss some of the properties of locally Lipschitz maps. In particular,
the image of a flat chain under a Lipschitz mapping is examined. In
addition, Lipschitz embeddings and the properties of the set they
constitute are considered. This enables the presentation of a Lipschitz
type configuration space in Chapter \ref{sec:configuration-space-and}.
In Chapter \ref{sec:Sharp_functions} we discuss the product of locally
Lipschitz maps and flat chains. This multiplication operation is used
in the definition of a local virtual velocity. As an example for the
use this proposed setting, Reynolds  transport theorem is presented
in Chapter \ref{chap:Reynolds-transport-theorem}. Our main theorem
is presented in Chapter \ref{sec:Cauchy-fluxes} where we prove that
a system of forces obeying balance and weak balance is equivalent
to a unique $n$-tuple of flat $(n-1)$-cochains. Generalized bodies
and surfaces are introduced in Chapters \ref{sec:Generalized-bodies}.
Virtual strains, or velocity gradients, stresses and a generalized
form of the principle of virtual work are presented in Chapters \ref{sec:virtual-strains-and_virtual_work}
and \ref{sec:Stress}.

%% file: gmt.tex
\chapter{Elements of Geometric Measure Theory\label{sec:Federer}}

In this chapter, some of the fundamental concepts form the theory
of currents in $\reals^{n}$ are presented. Throughout, the notation
is mainly adopted from \cite[Chapter 4]{Federer1969}. The notion
of flat forms needed for Wolfe's representation theorem, originally
presented in Whitney's Geometric Integration Theory \cite[Chapter VII]{Whitney1957},
is formulated in this section by the tools of Federer's Geometric
Measure Theory.

Let $\oset$ be an open set in $\reals^{n}$ and $V$ a vector space.
The notation $\D^{m}\left(\oset,V\right)$ is used for the vector
space of smooth, compactly supported $V$-valued differential $m$-forms
defined on $\oset$ and $\D^{m}\left(\oset\right)$ is used as an
abbreviation for $\D^{m}\left(\oset,\reals\right)$. The notation
$\cbnd\form$ is use for the \textit{exterior derivative }\textit{\emph{of}}
$\form\in\D^{m}(\oset)$, an element of $\D^{m+1}(\oset)$. The vector
space $\D^{m}(\oset)$ will be endowed with a locally convex topology
induced by a family of semi-norms \cite[p.~344]{Federer1969} as in
the theory of distributions.

A continuous linear functional $T:\D^{m}(U)\to\reals$ is referred
to as an $m$-\textit{dimensional current} in $\oset$. The collection
of all $m$-dimensional currents defined on $\oset$ forms the vector
space $\D_{m}(U)$ which is the vector space dual to $\D^{m}(\oset)$.
Let $T\in\D_{m}(\oset)$ with $m\geq1$ then $\bnd T$, the\textit{
boundary} of $T$ is the element of $\D_{m-1}(U)$ defined by 
\begin{equation}
\bnd T(\form)=T(\cbnd\form),\quad\text{for all}\quad\form\in\D^{m-1}(\oset).
\end{equation}
The support of a current $T\in\D_{m}(\oset)$ is defined by 
\[
\spt\left(T\right)=\oset\backslash\left\{ W\mid W\text{-is open},\, T(\phi)=0\,\,\text{for all }\phi\in\D^{m}(\oset),\,\spt(\phi)\subset W\right\} .
\]
Generally speaking, the support of a current $T\in\D_{m}\left(\oset\right)$
need not be compact. However, we note that all currents introduced
in this work will be of compact support.

The exterior derivative $\cbnd$ is a continuous linear map $\cbnd:\D^{m}(\oset)\to\D^{m+1}(\oset)$.
Thus, the boundary operation $\bnd:\D_{m+1}(\oset)\to\D_{m}(\oset)$,
viewed as the adjoint operator to $\cbnd$, is a continuous linear
operator on currents.

As an example of a $0$-current in $\oset$, let $\lusb^{n}$, denote
the $n$-dimensional Lebesgue measure in $\reals^{n}$. Then, the
restricted measure $\lusb^{n}\rest\oset$ is the $0$-current defined
as 
\begin{equation}
\lusb^{n}\rest\oset(\form)=\int_{\oset}\form d\lusb^{n},\quad\text{for all}\quad\form\in\D^{0}(\oset).\label{eq:L^n_0_current}
\end{equation}
 Given $\eta$, a Lebesgue integrable $m$-vector field defined on
$\oset$, then, $\lusb^{n}\wedge\eta$ denotes the $m$-current in
$\oset$ defined by 
\begin{equation}
\lusb^{n}\wedge\eta(\form)=\int_{\oset}\form(\eta)d\lusb^{n},\quad\text{for all}\quad\form\in\D^{m}(\oset).\label{eq:L^n_m_current}
\end{equation}

The inner product in $\reals^{n}$ induces an inner product in $\bigwedge_{m}\reals^{n}$
and $\mass{\xi}$ will denote the resulting norm of an $m$-vector
$\xi$. If $\xi$ is an $m$-vector given by $\xi=\sum_{i}\xi^{i}E_{i}$
with $\left\{ E_{i}\right\} $ the standard base for $\bigwedge_{m}\reals^{n}$
such that $i=1,\dots,\left(\begin{array}{c}
n\\
m
\end{array}\right)$ it follows that 
\begin{equation}
\mass{\xi}=\sqrt{\left\langle \xi,\xi\right\rangle }=\sqrt{\left(\xi^{i}\right)^{2}}.
\end{equation}
For an $m$-covector $\alpha$, define $\|\alpha\|$ by 
\begin{equation}
\|\alpha\|=\sup\left\{ \alpha(\xi)\mid\mass{\xi}\leq1,\,\,\xi\text{ is a simple }m\text{-vector}\right\} .
\end{equation}
Dually, for an $m$-vector $\xi$, $\|\xi\|$ is defined by 
\begin{equation}
\|\xi\|=\sup\left\{ \phi(\xi)\mid\phi\in\bigwedge^{m}\reals^{n},\,\,\|\phi\|\leq1\right\} ,
\end{equation}
which results in 
\begin{equation}
\|\xi\|=\inf\left\{ \sum\|\xi_{i}\|\mid\sum\xi_{i}=\xi,\quad\xi_{i}-\text{simple \ensuremath{m}-vector}\right\} .
\end{equation}
Given $\form\in\D^{m}(\oset)$, for every $x\in\oset$, $\form(x)$
is an $m$-covector, and so  
\begin{equation}
\|\form(x)\|=\sup\left\{ \form(x)(\xi)\mid\mass{\xi}\leq1,\,\,\xi\text{ is a simple }m\text{-vector}\right\} .
\end{equation}
 The \textit{comass }\textit{\emph{of}}\textit{ $\form$} is defined
by 
\begin{equation}
\Fmass{\form}=\sup_{x\in\oset}\|\form(x)\|.\label{eq:Mass_form}
\end{equation}
For $T\in\D_{m}(\oset)$ the \textit{mass }\textit{\emph{of}}\textit{
$T$} is dually defined by 
\begin{equation}
\Fmass T=\sup\left\{ T\left(\form\right)\mid\phi\in\D^{m}\left(\oset\right),\,\,\Fmass{\form}\leq1\right\} .\label{eq:Mass_current}
\end{equation}

An $m$-dimensional current $T$ is said to be \textit{represented
by integration} if there exists a Radon measure $\measure T$ and
an $m$-vector valued, $\measure T$-measurable function, $\vec{T}$,
with $\mass{\vec{T}(x)}=1$ for $\measure T$-almost all $x\in\oset$,
such that 
\begin{equation}
T\left(\form\right)=\int_{\oset}\form(\vec{T})d\measure T,\quad\text{for all}\quad\form\in\D^{m}(\oset).\label{eq:current_by_integration}
\end{equation}
A sufficient condition for an $m$-dimensional current, $T$, to be
represented by integration is that $T$ is a current of finite mass,
\textit{i.e.}, $\Fmass T<\infty$. An $m$-current $T$ is said to
be \textit{locally normal} if both $T$ and $\bnd T$ are represented
by integration and is said to be a \textit{normal} current if it is
locally normal and of compact support. The notion of normal currents
leads to the definition 
\begin{equation}
\Fnormal T=\Fmass T+\Fmass{\bnd T},\label{eq:N_norm}
\end{equation}
and clearly, every $T\in\D_{m}(\oset)$ such that $\Fnormal T<\infty$
is a normal current. The vector space of all $m$-dimensional normal
currents in $\oset$ is denoted by $N_{m}\left(\oset\right)$. For
a compact set $\compact$ of $\oset$, set 
\begin{equation}
N_{m,\compact}\left(\oset\right)=N_{m}(\oset)\cap\left\{ T\mid\spt\left(T\right)\subset\compact\right\} .\label{eq:N_km}
\end{equation}

For each compact subset $\compact$ of $\oset$, define $F_{K}$,
the $K$-\textit{flat semi-norm} on $\D^{m}\left(\oset\right)$, by
\begin{equation}
F_{K}\left(\form\right)=\sup_{x\in\compact}\left\{ \|\form(x)\|,\|\cbnd\form(x)\|\right\} .\label{eq:Flat_norm_forms}
\end{equation}
Dually, the \textit{$K$-flat norm }for currents $T\in\D_{m}\left(\oset\right)$
is given by 
\begin{equation}
F_{\compact}(T)=\sup\left\{ T\left(\form\right)\mid F_{\compact}\left(\form\right)\leq1\right\} .\label{eq:Flat_norm_currents}
\end{equation}
Note that if $T\in\D_{m}\left(\oset\right)$ such that $F_{\compact}(T)<\infty$,
then, $\spt(T)\subset\compact$. For a given compact subset $\compact\subset\oset$,
the set $F_{m,\compact}(U)$ is defined as the $F_{\compact}$ closure
of $N_{m,\compact}(U)$ in $\D_{m}(U)$. In addition, set 
\begin{equation}
F_{m}(\oset)=\bigcup_{\compact}F_{m,\compact}(U),\label{eq:F_m(U)}
\end{equation}
 where the union is taken over all compact subsets $\compact$ of
$\oset$. An element in $F_{m}(\oset)$ is referred to as a \textit{flat
$m$-chain in $\oset$}. 

For $T\in\D_{m}(\oset)$ with $\spt(T)\subset\compact$ it can be
shown that $F_{\compact}(T)$ is given by 
\begin{equation}
F_{\compact}(T)=\inf\left\{ \Fmass{T-\bnd S}+\Fmass S\mid S\in\D_{m+1}(\oset),\,\spt(S)\subset\compact\right\} .\label{eq:Flat_norm_Whitney}
\end{equation}
By taking $S=0$ we note that 
\begin{equation}
F_{\compact}(T)\leq\Fmass T.\label{eq:F<M}
\end{equation}
In addition, any element $T\in F_{m,\compact}\left(\oset\right)$
may be represented by $T=R+\bnd S$ where $R\in\D_{m}(\oset)$, $S\in\D_{m+1}(\oset)$,
such that $\spt(R)\subset\compact$, $\spt(S)\subset\compact$, and
\begin{equation}
F_{\compact}(T)=\Fmass R+\Fmass S.\label{eq:T=00003DR+bndS}
\end{equation}
Flat chains have some desirable properties. We note that the boundary
of a flat $m$-chain is a flat $(m-1)$-chain. Moreover, as Section
\ref{sec:On-Lipschitz-mappings} will show, the flat topology is preserved
under Lipschitz maps. From a geometric point of view the notion of
a flat chain may be used to describe objects of irregular geometric
nature such as the Sierpinski triangle. The following representation
theorem reveals the measure theoretic regularity characterization
of flat $m$-chains. 
\begin{thm}
\cite[Section 4.1.18]{Federer1969} Let $T$ be a flat $m$-chain
in $\oset$ with $\spt(T)\subset\compact$. Then, for any $\delta>0$
and $E=\left\{ x\mid\mathrm{dist}\left(\compact,x\right)\leq\delta\right\} \subset\oset$,
the current $T$ may be represented by 
\begin{equation}
T=\lusb^{n}\wedge\eta+\bnd\left(\lusb^{n}\wedge\xi\right),
\end{equation}
such that $\eta$ is an $\lusb^{n}\rest\oset$-summable, $m$-vector
field, $\xi$ is a $\lusb^{n}\rest\oset$-summable $\left(m+1\right)$-vector
field and $\spt\left(\eta\right)\cup\spt\left(\xi\right)\subset E$.\label{thm:Federer_representation_flat_chains}
\end{thm}
A linear functional $\cochain$ defined on $F_{m}(\oset)$ such that
there exists $0<c<\infty$ with $\cochain(T)\leq cF_{\compact}(T)$
for any compact $K\subset U$ and $T\in F_{m,\compact}(\oset)$, is
referred to as a \textit{flat $m$-cochain}. The flat norm of a cochains
is given by 
\begin{equation}
F(\cochain)=\sup\left\{ \cochain(\chain)\mid\chain\in F_{m}(\oset),\,\, F_{\compact}(\chain)\leq1,\,\,\compact\subset\oset\right\} .
\end{equation}
By Theorem \ref{thm:Federer_representation_flat_chains}, a dual representation
for flat cochains is available by flat forms which we shall now introduce.

Given a differentiable mapping $u$ defined on an open set of $\reals^{n}$,
its derivative will be denoted by $Du$ and its partial derivative
with respect to the $j$-th coordinate will be denoted by $D_{j}u$.
For a smooth $m$-vector field $\eta$ in $\oset$, \textit{the divergence
}$\mathrm{div}\eta$ \textit{\emph{of}}\textit{ $\eta$ }is an $\left(m-1\right)$-vector
field in $\oset$ defined by 
\begin{equation}
\mathrm{div}\eta=\sum_{j=1}^{n}D_{j}\eta\rest dx_{j},\label{eq:Div_definition}
\end{equation}
where $dx_{i},\: i=1,\dots,n$ denote the dual base vectors relative
to the standard basis $e_{j},\: j=1,\dots,n$ in $\reals^{n}$ \cite[Section 4.1.6]{Federer1969}.
For an integrable $m$-form $\phi$ in $\oset$, the \textit{weak
exterior derivative }\textit{\emph{of}}\textit{ $\form$} is defined
as an $\left(m+1\right)$ form in $\oset$ denoted by $\wcbd\form$
and such that the equality 
\begin{equation}
\int_{\oset}\wcbd\form(\eta)d\lusb^{n}=-\int_{\oset}\form\left(\mathrm{div}\eta\right)d\lusb^{n},\label{eq:Weak_exterior_driv_def}
\end{equation}
holds for all compactly supported, smooth $(m+1)$-vector fields $\eta$
on $\oset$. The weak exterior derivative is simply the exterior derivative
taken in the distributional sense. Note that $\wcbd\form$ is uniquely
defined up to a set of $\lusb^{n}\rest\oset$-measure zero, thus,
for $\form\in\D^{m}(\oset)$, the relation $\wcbd\form=\cbnd\form$
holds $\lusb^{n}\rest\oset$-almost everywhere. 

Differential forms whose components are Lipschitz continuous are referred
to as \textit{sharp $m$-forms} (adopting Whitney's terminology \cite[Section V.10]{Whitney1957}).
By Rademacher's theorem, the exterior derivative for sharp $m$-forms
exists $\lusb^{n}\rest\oset$-almost everywhere and the existence
of the weak exterior derivative follows. Sharp forms are clearly a
generalization of the notion of a smooth differential form and a further
generalization is given by flat forms where the Lipschitz continuity
is relaxed.
\begin{defn}
An $m$-form $\form$ in $\oset$ is said to be \textit{flat }if\label{def:flat_forms}

\begin{equation}
F(\form)=\sup_{\eta,\xi}\left\{ \int_{\oset}\left(\phi(\eta)+\wcbd\form(\xi)\right)d\lusb^{n}\right\} <\infty,\label{eq:Flat_form_def}
\end{equation}
where $\eta$ and $\xi$ are respectively $m$ and $(m+1)$ compactly
supported, $\lusb^{n}\rest\oset$-summable vector fields such that
\begin{equation}
\int_{\oset}\left(\|\xi\|+\|\eta\|\right)d\lusb^{n}=1.
\end{equation}

It is further observed that for $\form$, a flat $m$-form in $\oset$,
\begin{equation}
F(\form)=\ess\sup_{x\in\oset}\left\{ \|\form(x)\|,\|\wcbd\form(x)\|\right\} .\label{eq:Flat_form_def_result}
\end{equation}

Alternative definitions for flat forms may be found in \cite[Section IX.7]{Whitney1957}
and \cite[Section 5.5]{Heinonen2005}.\end{defn}
\begin{rem}
For $\form$, a flat $m$-form in $\oset$, and $\omega$, a flat
$r$-form in $\oset$, $\form\wedge\omega$ is a flat $(m+r)$-form
in $\oset$ for which we now examine the weak exterior derivative
$\wcbd(\form\wedge\omega)$. Let $\eta$ be a compactly supported
smooth $m$-vector field and $\xi$ a compactly supported smooth $r$-vector
field then
\begin{equation}
\begin{split}\int_{\oset}\wcbd(\form\wedge\omega)(\eta\wedge\xi)d\lusb^{n} & =\int_{\oset}(\form\wedge\omega)\left(\mathrm{div}\left(\eta\wedge\xi\right)\right)d\lusb^{n},\\
 & =\int_{\oset}(\form\wedge\omega)\left(\left(\mathrm{div}\eta\right)\wedge\xi+(-1)^{m}\eta\wedge\mathrm{div}\xi\right)d\lusb^{n},\\
 & =\int_{\oset}\left(\form(\mathrm{div}\eta)\omega(\xi)+(-1)^{m}\phi(\eta)\omega\left(\mathrm{div}\xi\right)\right)d\lusb^{n},\\
 & =\int_{\oset}(\wcbd\form(\eta)\omega(\xi)+(-1)^{m}\phi(\eta)\wcbd\omega\left(\xi\right))d\lusb^{n},\\
 & =\int_{\oset}\left(\wcbd\form\wedge\omega+(-1)^{m}\phi\wedge\wcbd\omega\right)\left(\eta\wedge\xi\right)d\lusb^{n}.
\end{split}
\end{equation}
Thus, 
\begin{equation}
\wcbd(\form\wedge\omega)=\wcbd\form\wedge\omega+\left(-1\right)^{m}\form\wedge\wcbd\omega.\label{eq:exterior_deriv_wedge_flat_form}
\end{equation}
which is a generalization of the well known analogous formula for
the exterior derivative of the exterior product of smooth forms.
\end{rem}
The representation theorem of flat cochains is traditionally referred
to as Wolfe's representation theorem, \cite[Chapter IX]{Whitney1957},
\cite[Section 4.1.19]{Federer1969}. It states that any flat $m$-cochain
$X$ in $\oset$ is represented by a flat $m$-form denoted by $\fform{\cochain}$
such that 
\begin{eqnarray}
\cochain\left(\lusb^{n}\wedge\eta+\bnd\left(\lusb^{n}\wedge\xi\right)\right) & = & \int_{\oset}\left[\fform{\cochain}(\eta)+\wcbd\fform{\cochain}(\xi)\right]d\lusb^{n},\label{eq:Wolfe's_representation}
\end{eqnarray}
for any $\eta$ and $\xi$, compactly supported, $\lusb^{n}\rest\oset$-summable
$m$ and $(m+1)$-vector fields, respectively. It is further noted
that the flat norm $F(\cochain)$ for the cochain $X$ is given by
\begin{eqnarray}
F(\cochain) & = & \mathrm{ess}\sup_{x\in\oset}\left\{ \|\fform{\cochain}(x)\|,\|\wcbd\fform{\cochain}(x)\|\right\} \equiv F(\fform{\cochain}).\label{eq:F_norm_cochain}
\end{eqnarray}
The \textit{coboundary} of a flat $m$-cochain $\cochain$ is defined
as the flat $(m+1)$-cochain $\cbnd\cochain$ such that 
\begin{equation}
\cbnd\cochain(\chain)=\cochain(\bnd\chain),\quad\text{for all}\quad\chain\in F_{m}(\oset),\label{eq:coboundary_cochain}
\end{equation}
where it is noted that the same notation is used for the coboundary
operator and the exterior derivative. The coboundary is the adjoint
of the boundary operator and thus a continuous linear operator taking
flat $m$-chains to flat $(m+1)$-chains. It follows from the representation
theorem of flat chains that the flat $(m+1)$-cochain $\cbnd\cochain$
is represented by the flat $(m+1)$-form $\fform{\cbnd\cochain}=\wcbd\fform{\cochain}$.
The last equality is used as the definition of the exterior derivative
of a flat form in \cite[Section IX.12]{Whitney1957}.

Given a flat $m$-cochain $\cochain$ in $\oset$ and a flat $r$-cochain
$Y$ in $\oset$, then, $\cochain\wedge Y$ is an $(m+r)$-cochain
represented by the flat $(m+r)$-form $\fform{\cochain\wedge Y}=\fform{\cochain}\wedge\fform Y$,
and for a flat $\left(m+r\right)$-chain $T=\lusb^{n}\wedge\eta+\bnd\left(\lusb^{n}\wedge\xi\right)$,
$\cochain\wedge Y(T)$ is defined by Equation(\ref{eq:Wolfe's_representation}).
Moreover, Equation (\ref{eq:exterior_deriv_wedge_flat_form}) implies
that 
\begin{equation}
\cbnd(\cochain\wedge Y)=\cbnd\cochain\wedge Y+(-1)^{m}\cochain\wedge\cbnd Y.\label{eq:coboundary_of_wedge}
\end{equation}

For a flat $m$-cochain $X$ and a flat $r$-chain $T$, such that
$m\leq r$, the interior product $\cochain\irest T$ is defined as
a flat $\left(r-m\right)$-chain such that 
\begin{equation}
\cochain\irest T(\omega)=\left(\cochain\wedge\omega\right)(T),\quad\text{for all}\quad\omega\in\D^{r-m}(\oset),\label{eq:interior_product_chain}
\end{equation}
where $\cochain\wedge\omega$ is the flat $r$-cochain represented
by the flat $r$-form $\fform{\cochain}\wedge\omega$.

%% file: Sets_of_finite_perimeter.tex
\chapter[Sets of Finite Perimeter]{Sets of finite perimeter, bodies and material surfaces\label{sec:Sets_of_finite_perimeter}}

In this chapter we lay down the basic assumptions regarding the collection
of admissible bodies. Sets of finite perimeter, or Caccioppoli sets,
will play a central role in the proposed framework. We first recall
some of the properties of sets of finite perimeter. Extended presentations
 of the subject may be found in \cite{Federer1969,Ziemer1983,Ziemer1989,Evans1992}. 

Let $U$ be a Borel set in an open subset of $\reals^{n}$ and $B(x,r)$
be the ball centered at $x\in\reals^{n}$ with radius $r$. Define\textit{
the $\oset$-density of the point $x$} by 
\begin{equation}
\dns xU=\lim_{r\to0}\frac{\lusb^{n}\left(U\cap\ball xr\right)}{\lusb^{n}\left(\ball xr\right)},\label{eq:U_density}
\end{equation}
where the limit exists. For $\alpha\in[0,1]$ set 
\begin{equation}
\oset^{\alpha}=\left\{ x\in\reals^{n}\mid d\left(x,\oset\right)=\alpha\right\} .\label{eq:U^alpha}
\end{equation}
The measure theoretic boundary, $\mtb\left(U\right)$, of the set
$U$ is defined by 
\begin{equation}
\mtb\left(U\right)=\reals^{n}-\left(\oset^{1}\cup\oset^{0}\right).\label{eq:measure_theoretic_boundary_def}
\end{equation}
 
\begin{defn}
A Borel set $U$ in $\reals^{n}$ is said to be \textit{a set of finite
perimeter} if $\lusb^{n}\left(U\right)<\infty$ and $H^{n-1}\left(\mtb\left(U\right)\right)<\infty$,
where $H^{n-1}\left(\mtb\left(U\right)\right)$ is the $(n-1)$-Hausdorff
measure of $\mtb\left(U\right)$.\label{def:Set-of-finite-perimeter}
\end{defn}
For $x\in\reals^{n}$, let $\nu_{\oset}\left(x\right)$ denote a unit
vector in $\reals^{n}$ and define 
\begin{equation}
\begin{split}B^{+}(x,r) & =B(x,r)\cap\left\{ y\mid(y-x)\cdot\nu_{\oset}\left(x\right)\geq0\right\} ,\\
B^{-}(x,r) & =B(x,r)\cap\left\{ y\mid(y-x)\cdot\nu_{\oset}\left(x\right)\leq0\right\} .
\end{split}
\end{equation}
The vector $\nu\left(x,\oset\right)$ is said to be the \textit{measure
theoretic exterior normal} to $\oset$ at $x$ if 
\begin{equation}
\lim_{r\to0}\frac{\lusb^{n}\left(B^{+}\left(x,r\right)\cap\oset\right)}{\lusb^{n}\left(B^{+}\left(x,r\right)\right)}=0,
\end{equation}
and 
\begin{equation}
\lim_{r\to0}\frac{\lusb^{n}\left(B^{-}\left(x,r\right)\cap\oset\right)}{\lusb^{n}\left(B^{-}\left(x,r\right)\right)}=1.
\end{equation}
For a set of finite perimeter, the exterior normal $\nu_{\oset}\left(x\right)$
to $U$ exists $H^{n-1}$-almost everywhere in $\mtb(\oset)$ thus
making a generalized version of the Gauss-Green theorem applicable.

Several equivalent definitions for a set of finite perimeter may be
found in the literature. In \cite[Section 5.4.1]{Ziemer1989} a set
of finite perimeter is viewed as a set $\oset$ whose characteristic
function $\chi_{\oset}:\reals^{n}\to\reals$ defined by 
\begin{equation}
\chi_{\oset}(x)=\begin{cases}
1, & x\in\oset,\\
0, & x\not\in\oset,
\end{cases}
\end{equation}
 is a function of bounded variation in $\reals^{n}$. Let $f$ be
a real valued function defined on the open set $\oset$. The \textit{total
variation} of $f$ is defined by 
\begin{equation}
\|\var f\|=\sup\left\{ \int_{\oset}f\cdot\mathrm{div}(\varphi)d\lusb^{n}\mid\varphi\in C_{0}^{\infty}\left(\oset,\reals^{n}\right),\varphi(x)\leq1\text{ for all }x\in\oset\right\} ,
\end{equation}
\lyxadded{Reuven Segev}{Thu Apr 25 11:29:17 2013}{}where $C_{0}^{\infty}\left(\oset,\reals^{n}\right)$
is used to denote the space of smooth, compactly supported $\reals^{n}$-valued
functions defined on $V$. A function $u\in\lusb^{1}\left(V\right)$
is said to be a\textit{ function of bounded variation} in $V$ if
each of partial derivatives $D_{i}u$ (taken in the distributional
sense) are Radon measures with a finite total variation. Alternatively,
a function $u\in\lusb^{1}\left(V\right)$ is a function of bounded
variation in $V$ if 
\begin{equation}
\norm{\oset}{BV(\oset)}=\int_{\oset}\mass fd\lusb^{n}+\|\var f\|<\infty.
\end{equation}
In this sense, the measure theoretic exterior normal is defined by
\begin{equation}
\nu_{\oset}(x)=\lim_{r\to\infty}-\frac{D\chi_{U}\left(B(x,r)\right)}{\|D\chi_{U}\|\left(B(x,r)\right)}.
\end{equation}
In \cite[Section 4.5]{Federer1969}, a set of finite perimeter is
viewed as a set $\oset$ such that the current $\left(\lusb^{n}\rest\oset\right)\wedge e_{1}\wedge\dots\wedge e_{n}$
is an integral $n$-current in $\reals^{n}$. In this work, Definition
\ref{def:Set-of-finite-perimeter} is chosen for its intuitive geometric
interpretation.

. Let $\body$ be an open set in $\reals^{n}$. A\textit{ body} in
$\body$ is denoted by $\part$ and is postulated to be a set of finite
perimeter in $\body$. Strictly speaking, a set of finite perimeter
is determined up to a set of $\lusb^{n}$ measure zero, thus as a
point set, it is not uniquely defined. Formally, each set of finite
perimeter determines an equivalence class of sets. A unique representation
of a body is given by the identification of the body $\part$ with
$T_{\part}$, an $n$-current in $\body$ defined as $T_{\part}=\left(\lusb^{n}\rest\part\right)\wedge e_{1}\wedge\dots\wedge e_{n}$.
By Equation (\ref{eq:L^n_m_current}), 
\begin{equation}
T_{\part}(\omega)=\int_{\part}\omega(x)\left(e_{1}\wedge\dots\wedge e_{n}\right)d\lusb_{x}^{n},\quad\text{for all}\quad\omega\in\D^{n}(\body).\label{eq:Current_T_part}
\end{equation}
Using the terminology of currents represented by integration, $\measure{T_{\part}}=\lusb^{n}\rest\part$
and $\vec{T}_{\part}=e_{1}\wedge\dots\wedge e_{n}$ are the Radon
measure and unit $n$-vector associated with the current $T_{\part}$.

Objects of dimension $(n-1)$ for which one can compute the flux will
be referred to as material surfaces. Formally, a \textit{material
surface} is defined as a pair $\surface=(\hat{\surface},v)$ where
$\hat{\surface}$ is a Borel subset of $\body$ such that for some
body $\part$ we have $\hat{\surface}\subset\mtb(\part)$ and $v$
is the exterior normal of $\part$ such that $v(x)=v_{\part}(x)$
is defined $H^{n-1}$-almost everywhere on $\hat{\surface}$. Let
$v^{*}(x)$ be a the covector defined by 
\begin{equation}
v^{*}(x)(u)=v(x)\cdot u,\quad\text{for all}\quad u\in\reals^{n},
\end{equation}
and set $\vec{T}_{\surface}$ as the $(n-1)$-vector 
\begin{equation}
\vec{T}_{\surface}(x)=v^{*}(x)\irest e_{1}\wedge\dots\wedge e_{n}.\label{eq:T_S_def_vec_def}
\end{equation}
It is easy to show that $\vec{T}_{\surface}(x)$ is a unit, simple
$(n-1)$-vector $H^{n-1}$-almost everywhere on $\hat{\surface}$.
 We use $T_{\surface}$ to denote the $\left(n-1\right)$-current
in $\body$ induced by the material surface $\surface$, such that
$\measure{T_{\surface}}=H^{n-1}\rest\hat{\surface}$ and $\vec{T}_{\surface}(x)$
are the Radon measure and $(n-1)$-vector associated with $T_{\surface}$,
and 
\begin{equation}
T_{\surface}(\omega)=\int_{\hat{\surface}}\omega(x)(\vec{T}_{\surface}(x))dH_{x}^{n-1},\quad\text{for all}\quad\omega\in\D^{n-1}(\body).\label{eq:T_S_current}
\end{equation}
The unit $(n-1)$-vector $\vec{T}_{\surface}(x)$ is viewed as the
natural $(n-1)$-vector \textit{tangent }\textit{\emph{to the material
surface}}\textit{ $\surface$}. By Equation (\ref{eq:L^n_m_current})
we may write 
\begin{equation}
T_{\surface}=\left(H^{n-1}\rest\hat{\surface}\right)\wedge\vec{T}_{\surface}.\label{eq:T_S_def}
\end{equation}

Consider the material surface $\bnd\part=\left(\mtb(\part),\nu_{\part}\right)$
naturally induced by the body $\part$. One has, 
\begin{equation}
\begin{split}T_{\bnd\part}(\omega) & =\int_{\mtb(\part)}\omega(x)(\vec{T}_{\bnd\part}(x))dH_{x}^{n-1},\\
 & =\int_{\mtb(\part)}\left(\omega(x)\irest e_{1}\wedge\dots\wedge e_{n}\right)\cdot\nu_{\part}(x)dH_{x}^{n-1},\\
 & =\int_{\part}\cbnd\omega(x)\left(e_{1}\wedge\dots\wedge e_{n}\right)d\lusb_{x}^{n},\\
 & =T_{\part}\left(\cbnd\omega\right),\\
 & =\bnd T_{\part}(\omega).
\end{split}
\label{eq:Stoke's_thm}
\end{equation}
where in the third line above Gauss-Green theorem \cite[Section 4.5.6]{Federer1969}
was used. Thus, it is noted that $T_{\bnd\part}=\bnd T_{\part}$ as
expected, and the material surface $\surface$ associated with the
body $\part$ may be written as 
\begin{equation}
T_{\surface}=\left(\bnd T_{\part}\right)\rest\hat{\surface}.\label{eq:T_S=00003Dbnd_T_P_res_S}
\end{equation}
Since a Radon measure is a Borel regular measure, the current $\bnd T_{\part}\rest\hat{\surface}$
is well defined for any Borel set $\hat{\surface}$ \cite[p. 356]{Federer1969}. 

For each $T_{\part}$, we observe that $\Fmass{T_{\part}}=\lusb^{n}\left(\part\right)$
and $\Fmass{\bnd T_{\part}}=H^{n-1}\left(\mtb(\part)\right)$ correspond
to the ``volume'' of the body and ``area'' of its boundary, respectively.
By Equation (\ref{eq:N_norm}) one has $\Fnormal{T_{\part}}=\lusb^{n}\left(\part\right)+H^{n-1}\left(\mtb(\part)\right)<\infty$,
so that the current $T_{\part}$ is a normal $n$-current in $\body$,
in particular $T_{\part}$ is an integral $n$-current. The open set
$\body$ is referred to as \textit{\emph{the }}\textit{universal body}
and we define\emph{ }\textit{\emph{the }}\textit{class of admissible
bodies,} $\Omega_{\body}$, as the collection of all bodies in the
universal body $\body$, \textit{i.e.}, 
\begin{equation}
\Omega_{\body}=\left\{ T_{\part}\mid\part\subset\body,T_{\part}=\lusb^{n}\rest\part\in N_{n}\left(\body\right)\right\} .
\end{equation}
The result obtained in \cite{M.E.Gurtin1986} implies that in case
$\body$ is assumed to be a set of finite perimeter, $\Omega_{\body}$
would have the structure of a Boolean algebra and would form a material
universe in the sense of Noll \cite{Noll1973}. In Section \ref{sec:Generalized-bodies},
a generalized class of admissible bodies will be defined for which
a requirement that $\body$ is a bounded set will be sufficient in
order to construct a Boolean algebra structure. 

The collection of all material surfaces in $\body$ will be denoted
by $\bnd\Omega_{\body}$, so that
\begin{equation}
\bnd\Omega_{\body}=\left\{ T_{\surface}\mid T_{\surface}=\left(\bnd T_{\part}\right)\rest\hat{\surface},T_{\part}\in\Omega_{\body}\right\} .
\end{equation}
By the definition of $T_{\surface}$ it follows that $\Fmass{T_{\surface}}=H^{n-1}\left(\hat{\surface}\right)$
for each $T_{\surface}\in\bnd\Omega_{\body}$. Thus $T_{\surface}$
is a flat $(n-1)$-chain of finite mass. The material surfaces $T_{\surface}$
and $T_{\surface'}$ are said to be \textit{compatible} if there exists
a body $T_{\part}$ such that $T_{\surface}=\left(\bnd T_{\part}\right)\rest\hat{\surface}$
and $T_{\surface'}=\left(\bnd T_{\part}\right)\rest\hat{\surface'}$.
The material surfaces $T_{\surface}$ and $T_{\surface'}$ are said
to be \emph{disjoint} if $\mathrm{clo}\bigl(\hat{\surface}\bigr)\cap\mathrm{clo}\bigl(\hat{\surface}'\bigr)=\varnothing.$

%% file: Lipschitz.tex
\chapter[Lipschitz maps and chains]{Lipschitz mappings and Lipschitz chains\label{sec:On-Lipschitz-mappings}}

Lipschitz mappings will model configurations of bodies in space. In
this chapter we review briefly some of their relevant properties.

A map $\Lmap:\oset\to V$ from an open set $\oset\subset\reals^{n}$
to an open set $V\subset\reals^{m}$, is said to be a\textit{ (globally)
Lipschitz map} if there exists a number $c<\infty$ such that $\mass{\Lmap(x)-\Lmap(y)}\leq c\mass{x-y}$
for all $x,\, y\in\oset$. The \textit{Lipschitz constant }of $\Lmap$
is defined by 
\begin{equation}
\Lip_{\Lmap}=\sup_{x,y\in\oset}\frac{|\Lmap(y)-\Lmap(x)|}{|y-x|}.\label{eq:Lipschitz-constant}
\end{equation}
The map $\Lmap:\oset\to V$ is said to be \textit{locally Lipschitz}
if for every $x\in\oset$ there is some neighborhood $\oset_{x}\subset\oset$
of $x$ such that the restricted map $\Lmap\mid_{\oset_{x}}$ is a
Lipschitz map.

Let $\Lmap:\oset\to\reals^{m}$ be a locally Lipschitz map defined
on the open set $\oset\subset\reals^{n}$, then for every $\compact$,
a compact subset of $\oset$, the restricted map $\Lmap\mid_{\compact}$
is globally Lipschitz in the sense that $\Lip_{\Lmap,\compact}$,
the $\compact$-Lipschitz constant of the map $\Lmap\mid_{\compact}$,
given by 
\begin{equation}
\Lip_{\Lmap,\compact}=\sup_{x,y\in\compact}\frac{\mass{\Lmap(x)-\Lmap(y)}}{\mass{x-y}},
\end{equation}
is finite.

\section{Differential topology of Lipschitz maps}

\sectionmark{}

The vector space of locally Lipschitz mappings from the open set $\oset\subset\reals^{n}$
to the open set $V\subset\reals^{m}$ is denoted by $\Lip\left(\oset,V\right)$.
For a compact subset $\compact\subset\oset$, define the semi-norm
\begin{equation}
\norm{\Lmap}{\Lip,\compact}=\max\left\{ \norm{\Lmap\mid_{\compact}}{\infty},\Lip_{\Lmap,\compact}\right\} ,\label{eq:Lipschit_semi_norm}
\end{equation}
 on $\Lip\left(\oset,V\right)$, where, 
\begin{equation}
\|\Lmap\mid_{\compact}\|_{\infty}=\sup_{x\in\compact}\mass{\Lmap(x)}.
\end{equation}
The vector space $\Lip(\oset,V)$ is endowed with the strong Lipschitz
topology (see \cite{Fukui2005}). It is the analogue of Whitney's
topology (strong topology) for the space of differentiable mappings
between open sets (see \cite[p.~35]{Hirsch}) and is defined as follows.
\begin{defn}
\label{Lipschitz_strong_topology}Given $\Lmap\in\Lip(\oset,V)$,
for some indexing set $\Lambda$, let $\mathcal{U}=\left\{ U_{\lambda}\right\} _{\lambda\in\Lambda}$
be an open, locally finite cover of $\oset\subset\reals^{n}$, and
$\mathcal{K}=\left\{ \compact_{\lambda}\right\} _{\lambda\in\Lambda}$
a family of compact subsets in $\oset$ such that $\compact_{\lambda}\subset U_{\lambda}$
and $\delta=\left\{ \delta_{\lambda}\right\} _{\lambda\in\Lambda}$
a family of positive numbers. A neighborhood $B^{\Lip}\left(\Lmap,\mathcal{U},\delta,\mathcal{K}\right)$
of $\Lmap$ in the strong topology is defined as the collection of
all $g\in\Lip\left(\oset,V\right)$ such that $\norm{\Lmap-g}{\Lip,\compact_{\lambda}}<\delta_{\lambda}$,
\emph{i.e.},
\begin{equation}
B^{\Lip}\left(\Lmap,\mathcal{U},\mathcal{K},\delta\right)=\left\{ g\in\Lip(\oset,V)\mid\:\norm{\Lmap-g}{\Lip,\compact_{\lambda}}<\delta_{\lambda},\,\lambda\in\Lambda\right\} .\label{eq:Lipschitz_strong_topology}
\end{equation}

A map $\emb:\oset\too V$, with $\oset\subset\reals^{n}$, $V\subset\reals^{m}$,
open sets such that $m\geq n$, is said to be a \textit{bi-Lipschitz
}map if there are numbers $0<c\leq d<\infty$, such that \cite[p.~78]{Heinonen2000}
\begin{equation}
c\leq\frac{|\emb(x)-\emb(y)|}{\mass{x-y}}\leq d,\qquad\text{for all}\quad x,\, y\in\oset,\; x\not=y.\label{eq:bi-Lipschitz_constant}
\end{equation}
Setting $L=\max\left\{ \frac{1}{c},d\right\} $, 
\begin{equation}
\frac{1}{L}\leq\frac{|\emb(x)-\emb(y)|}{\mass{x-y}}\leq L,\qquad\text{for all}\quad x,\, y\in\oset,\; x\not=y,
\end{equation}
and in such a case $\emb$ is said to be $L$-bi-Lipschitz. 

The map $\Lmap:\oset\to V$, where $\oset\subset\reals^{n}$ and $V\subset\reals^{m}$
are open sets such that $m\geq n$, is a \textit{Lipschitz immersion}
if for every $x\in\oset$ there is a neighborhood $U_{x}\subset\oset$
of $x$ such that $\Lmap\mid_{U_{x}}$ is a bi-Lipschitz map, \textit{i.e.},
there are $0<c_{x}\leq d_{x}<\infty$, and 
\begin{equation}
c_{x}\leq\frac{|\emb(y)-\emb(z)|}{\mass{y-z}}\leq d_{x},\quad\text{for all}\quad y,\, z\in U_{x},\; y\not=z.
\end{equation}
\end{defn}
\begin{lem}
The set of Lipschitz immersions is an open subset of $\Lip(\oset,V)$
with respect to the strong Lipschitz topology\label{lem:Immersion_open}.\end{lem}
\begin{proof}
Let $g:\oset\to V$ be a Lipschitz immersion and for $x\in\oset$,
let $\mathcal{U}_{x}$ be a bounded open set containing $x$ such
that $g|_{\mathcal{U}_{x}}$ is a bi-Lipschitz map. The collection
$\left\{ \mathcal{U}_{x}\right\} _{x\in\oset}$, forms an open cover
of $\oset$. Since $\reals^{n}$ is paracompact we may extract a locally
finite refinement $\mathcal{U}=\left\{ \mathcal{U}_{\lambda}\right\} _{\lambda\in\Lambda}$
which is an open subcover of $\oset$. For each $\mathcal{U}_{\lambda}$
select an open set $\mathcal{V}_{\lambda}$ such that $\mathrm{clo}\left(\mathcal{V}_{\lambda}\right)\subset\mathcal{U}_{\lambda}$
and that $\left\{ \mathcal{V}_{\lambda}\right\} _{\lambda\in\Lambda}$
is an open cover of $\oset$. Denote the sets $\mathrm{clo}\left(\mathcal{V}_{\lambda}\right)$
by $\mathcal{K}_{\lambda}$ and note that $\mathcal{K}=\left\{ \mathcal{K}_{\lambda}\right\} _{\lambda\in\Lambda}$
is a locally finite cover of $\oset$ and each $\mathcal{K}_{\lambda}$
is a compact set  with non empty interior. For an extended proof
of existence of $\mathcal{K}$ we refer to \cite[Section 41]{Munkres2000}.

Let $x\in\mathcal{K}_{\lambda}$ for some compact set $\mathcal{K}_{\lambda}$,
then $g|_{\mathcal{K}_{\lambda}}$ is a $L_{\lambda}$-bi-Lipschitz
map for some $0<L_{\lambda}<\infty$ and let $\Lmap\in\Lip\left(\oset,V\right)$
then since $\Lmap|_{\mathcal{K}_{\lambda}}$ is a Lipschitz map it
will suffice to show that $\Lmap|_{\mathcal{K}_{\lambda}}$ is an
injective map. For every $z,\, y\in\mathcal{K}_{\lambda}$, 
\[
0<\frac{\mass{\emr(z)-\emr(y)}}{\mass{z-y}}\leq\frac{\mass{\left(\emr-\Lmap\right)(z)-\left(\emr-\Lmap\right)(y)}}{\mass{z-y}}+\frac{\mass{\Lmap(z)-\Lmap(y)}}{\mass{z-y}},
\]
 hence 
\begin{eqnarray*}
\frac{\mass{\Lmap(z)-\Lmap(y)}}{\mass{z-y}} & \geq & \frac{\mass{\emr(z)-\emr(y)}}{\mass{z-y}}-\frac{\mass{\left(\emr-\Lmap\right)(z)-\left(\emr-\Lmap\right)(y)}}{\mass{z-y}},\\
 & \geq & \frac{\mass{\emr(z)-\emr(y)}}{\mass{z-y}}-\norm{\emr-\Lmap}{\Lip,\mathcal{K}_{\lambda}}.
\end{eqnarray*}
Taking the infimum over $z,\, y\in\mathcal{K}_{\lambda}$ on both
sides it follows that 
\[
\inf_{z,\, y\in\mathcal{K}_{\lambda}}\frac{\mass{\Lmap(z)-\Lmap(y)}}{\mass{z-y}}\geq\inf_{z,\, y\in\mathcal{K}_{\lambda}}\frac{\mass{g(z)-g(y)}}{\mass{z-y}}-\norm{g-\Lmap}{\Lip,\mathcal{K}_{\lambda}}.
\]
Setting $\delta_{\lambda}=\frac{1}{2}\inf_{z,\, y\in\mathcal{K}_{\lambda}}\frac{\mass{\emr(z)-\emr(y)}}{\mass{z-y}}=\frac{1}{2L_{\lambda}}$
it follows that 
\[
\inf_{z,\, y\in\mathcal{K}_{\lambda}}\frac{\mass{\Lmap(z)-\Lmap(y)}}{\mass{z-y}}\geq\frac{1}{2L_{\lambda}},
\]
hence $\Lmap|_{\mathcal{K}_{\lambda}}$ is an injective map. Since
every $x\in\oset$ is contained in some $\mathcal{K}_{\lambda}$ it
follows that $\Lmap$ is a Lipschitz immersion.
\end{proof}
The following theorem pertaining to the set of Lipschitz embeddings
is given in \cite{Fukui2005} for the setting of Lipschitz manifolds
and its proof is analogous to the case of differentiable mappings
as in \cite[p.~36--38]{Hirsch}.
\begin{defn}
A Lipschitz map $\emb:\oset\to V$ is said to be a\textit{ Lipschitz
embedding} if it is a Lipschitz immersion and a homeomorphism of $\oset$
onto $\emb(\oset)$.\end{defn}
\begin{thm}
The set $\Limb(\oset,V)$ is open in $\Lip(\oset,V)$ with respect
to the strong Lipschitz topology \label{thm:Lipschitz_embedding_open_set}\end{thm}
\begin{proof}
Let $\emb\in\Limb\left(\oset,V\right)$. Apply first the proof of
Lemma \ref{lem:Immersion_open} and obtain an open set $B_{0}^{\Lip}\left(\emb,\mathcal{U},\delta,\mathcal{K}\right)$
such that every element in $B_{0}^{\Lip}\left(\emb,\mathcal{U},\delta,\mathcal{K}\right)$
is a Lipschitz immersion. With $\mathcal{U}=\left\{ \mathcal{U}_{\lambda}\right\} _{\lambda\in\Lambda}$
and $\mathcal{K}=\left\{ \mathcal{K}_{\lambda}\right\} _{\lambda\in\Lambda}$
selected as in the proof of Lemma \ref{lem:Immersion_open}, recall
that $\emb|_{\mathcal{U}_{\lambda}}$ is an $L_{\lambda}$-bi-Lipschitz
map. Set 
\[
h_{\lambda}=\mathrm{dist}\left(\mathcal{K}_{\lambda},U-\mathcal{U}_{\lambda}\right)=\inf\left\{ \mass{y-z}\mid y\in\mathcal{K}_{\lambda},z\in U-\mathcal{U}_{\lambda}\right\} ,
\]
and $\delta'_{\lambda}=\frac{h_{\lambda}}{L_{\lambda}}$. Note that
for $\delta'=\left\{ \delta'_{\lambda}\right\} _{\lambda\in\Lambda}$,
the open set $B_{1}^{\Lip}\left(\emb,\mathcal{U},\delta',\mathcal{K}\right)$
contains the collection of elements $g\in\Lip\left(U,V\right)$ such
that for any $\mathcal{U}_{\lambda}\in\mathcal{U}$ and $\mathcal{K}_{\lambda}\in\mathcal{K}$
$g\left(\mathcal{K}_{\lambda}\right)\subset\emb\left(\mathcal{U}_{\lambda}\right)$.
Since $\emb$ is an embedding for any $\lambda\in\Lambda$ we may
find disjoint open sets $\mathcal{A}_{\lambda},\mathcal{B}_{\lambda}$
in $V$ such that
\[
\emb\left(\mathcal{V}_{\lambda}\right)\subset\mathcal{A}_{\lambda},\quad\emb\left(\oset-\mathcal{K}_{\lambda}\right)\subset\mathcal{B}_{\lambda},
\]
where $\mathcal{V}_{\lambda}$ is the open set satisfying $\mathrm{clo}\left(\mathcal{V}_{\lambda}\right)=\mathcal{K}_{\lambda}$.
Letting $f\in B_{0}^{\Lip}\left(\emb,\mathcal{U},\delta,\mathcal{K}\right)\cap B_{1}^{\Lip}\left(\emb,\mathcal{U},\delta',\mathcal{K}\right)$,
we now show that $f$ is injective. Let $x\in\mathcal{V}_{\lambda}$
and $y\in\oset$. If $y\in\mathcal{K}_{\lambda}$ then $f(x)\not=f(y)$
since $f|_{\mathcal{K}_{\lambda}}$ is a Lipschitz immersion. In case
$y\in\oset-\mathcal{K}_{\lambda}$, then $f(y)\in B_{\lambda}$ while
$f(x)\in\mathcal{A}_{\lambda}$. thus $f(x)\not=f(y)$ and $f$ is
injective.
\end{proof}

\section{Maps of currents induced by Lipschitz maps}

\sectionmark{}

Since our objective is to represent bodies as currents, and in particular,
as flat chains, and since we wish to represent configurations as Lipschitz
mappings, we exhibit in the following the basic properties of the
images of currents and chains under Lipschitz mappings.

Let $T$ be a current on $U$ and for open sets $\oset\subset\reals^{n}$
and $V\subset\reals^{m}$, let $\Lmap:\oset\too V$ be a smooth map
whose restriction to $\spt(T)$ is a proper map. For any $r$-form
$\omega$ on $V$, the map $\Lmap$ induces a form $\Lmap^{\#}\left(\omega\right)$,
the \textit{pullback of $\omega$ by $\Lmap$,} defined pointwise
by 
\begin{equation}
\left(\Lmap^{\#}\left(\omega\right)(x)\right)(v_{1}\wedge\dots\wedge v_{r})=\left(\omega\left(\Lmap\left(x\right)\right)\right)\left(D\Lmap(v_{1})\wedge\dots\wedge D\Lmap(v_{r})\right),\label{eq:pullback_flat_form}
\end{equation}
for all $v_{1},\dots v_{r}\in\reals^{n}$. It is observed that since
$\Lmap$ is proper only on $\spt(T)$, for a form $\omega$ with a
compact support, $\spt(\Lmap^{\#}(\omega))$ need not be compact.
However, for a real valued function $\zeta$ defined on $\oset$ which
is compactly supported and $\zeta(x)=1$ for all $x$ in a neighborhood
of $\spt\left(T\right)\cap\spt(\Lmap^{\#}(\omega))$, the smooth form
$\zeta\Lmap^{\#}\left(\omega\right)$ is of compact support. Thus,
the \textit{pushforward }$\Lmap_{\#}\left(T\right)$ \textit{\emph{of
$T$ by}}\emph{ }$\Lmap$ may be defined as the current in $V$ given
by 
\begin{equation}
\Lmap_{\#}\left(T\right)\left(\omega\right)=T\left(\zeta\Lmap^{\#}\left(\omega\right)\right),\quad\text{for all}\quad\omega\in\D^{r}\left(V\right),
\end{equation}
for any $\zeta$ with the properties given above \cite[Section 2.3]{Giaquinta1998}.
The definition of $\Lmap_{\#}\left(T\right)(\omega)$ is independent
of $\zeta$ and thus will be omitted in the following.  The pushforward
operation satisfies
\begin{eqnarray}
\bnd\Lmap_{\#}\left(T\right) & = & \Lmap_{\#}\left(\bnd T\right),\label{eq:bndF=00003DFbnd}\\
\spt\left(\Lmap_{\#}T\right) & \subset & \Lmap\left\{ \spt\left(T\right)\right\} .\label{eq:spt(FT)=00003DFspt(T)}
\end{eqnarray}
By a direct calculation one obtains that 

\begin{equation}
\Fmass{\Lmap_{\#}\left(T\right)}\leq\left(\sup_{x\in\compact}\mass{D\Lmap(x)}\right)^{r}\Fmass T.
\end{equation}
Applying Equation (\ref{eq:N_norm}) it follows that
\begin{equation}
N\left(\Lmap_{\#}\left(T\right)\right)\leq N(T)\sup\left\{ \left(\sup_{x\in\compact}\mass{D\Lmap(x)}\right)^{r},\left(\sup_{x\in\compact}\mass{D\Lmap(x)}\right)^{r-1}\right\} ,
\end{equation}
and by Equation (\ref{eq:T=00003DR+bndS}),
\begin{equation}
F_{\Lmap\left\{ \compact\right\} }\left(\Lmap_{\#}\left(T\right)\right)\leq F_{\compact}(T)\sup\left\{ \left(\sup_{x\in\compact}\mass{D\Lmap(x)}\right)^{r},\left(\sup_{x\in\compact}\mass{D\Lmap(x)}\right)^{r+1}\right\} ,
\end{equation}
 where $\Lmap\left\{ \compact\right\} $ is the image of the set $\compact$
under the map $\Lmap$.

In case $\Lmap:\oset\too V$ is a locally Lipschitz map, the map $\Lmap_{\#}$
cannot be defined as in the case of smooth maps. However, given any
compact $K\subset U$, for $T\in F_{r,\compact}\left(U\right)$, one
may define the current $\Lmap_{\#}\left(T\right)$ as a weak limit.

Let $\left\{ \Lmap_{\tau}\right\} $, $\tau\in\reals^{+}$, be a family
of smooth approximations of $\Lmap$ obtained by mollifiers \cite[Section 4.1.2]{Federer1969}.
(It is observed that flat chains have compact supports so that it
is not necessary to require that $\Lmap$ is proper.) Set 
\[
\Lmap_{\#}T(\omega)=\lim_{\tau\to0}\Lmap_{\tau\#}T(\omega),\quad\text{for all}\quad\omega\in\D^{r}(V).
\]
 The sequence $\left\{ \Lmap_{\tau\#}\left(T\right)\right\} $ is
a Cauchy sequence with respect to the flat norm  so that the limit
is well defined and one may write 
\begin{equation}
\Lmap_{\#}\left(T\right)=\lim_{\tau\to0}\Lmap_{\tau\#}\left(T\right).\label{eq:Lipschitz_image_of_chain}
\end{equation}
As a result, the locally Lipschitz map $\Lmap:U\to V$ induces a map
of flat chains 
\[
\Lmap_{\#}:F_{r}\left(U\right)\to F_{r}\left(V\right).
\]

Properties (\ref{eq:bndF=00003DFbnd}) and (\ref{eq:spt(FT)=00003DFspt(T)})
hold for the map $\Lmap_{\#}$ induced by a locally Lipschitz map
$\Lmap$ and 
\begin{equation}
\Fmass{\Lmap_{\#}\left(T\right)}\leq\Fmass T\left(\Lip_{\Lmap,\spt(T)}\right)^{r}.\label{eq:Lipschitz_mass_norm}
\end{equation}
It follows that for normal currents 
\begin{equation}
\begin{split}\Lmap_{\#}(T) & \in N_{r,\Lmap(K)}(V),\quad\text{for all}\quad T\in N_{r,K}(U),\\
N\left(\Lmap_{\#}\left(T\right)\right) & \leq N(T)\sup\left\{ \left(\Lip_{\Lmap,\spt(T)}\right)^{r},\left(\Lip_{\Lmap,\spt(T)}\right)^{r-1}\right\} ,
\end{split}
\label{eq:Lipschitz_normal_norm}
\end{equation}
and for flat chains 
\begin{equation}
\begin{split}\Lmap_{\#}(T) & \in F_{r,\Lmap\left\{ K\right\} }(V),\quad\text{for all}\quad T\in F_{r,K}(U),\\
F_{\Lmap\left\{ K\right\} }\left(\Lmap_{\#}\left(T\right)\right) & \leq F_{K}(T)\sup\left\{ \left(\Lip_{\Lmap,\spt(T)}\right)^{r},\left(\Lip_{\Lmap,\spt(T)}\right)^{r+1}\right\} .
\end{split}
\label{eq:Lipschitz_flat_norm}
\end{equation}
See \cite[Section 4.1.14]{Federer1969} and \cite[Section 2.3]{Giaquinta1998}
for an extended treatment. 

In Whitney's theory, the Lipschitz image of a flat chain $\chain$
is defined as follows \cite[Chapter X]{Whitney1957}. First, for $P=\spt(\chain)$
consider a full sequence of simplicial subdivision $\left\{ P_{i}\right\} $
such that $P_{i+1}$ is a simplicial refinement of $P_{i}$. Next,
let $\left\{ \Lmap_{i}\right\} $ be a sequence of piecewise affine
approximations of the Lipschitz map $\Lmap$ such that $\Lmap_{i}(v)=\Lmap(v)$
for all vertices $v$ in the simplicial complex $P_{i}$. The chain
$\Lmap_{\#}\left(\chain\right)$ is defined as the limit in the flat
norm of 
\begin{equation}
\Lmap(\chain)=\lim_{i\to\infty}\Lmap_{i}(\chain).
\end{equation}
Although Whitney's definition of $\Lmap_{\#}(\chain)$ differs from
that of Federer, the resulting chains are equivalent. 

For a locally Lipschitz map $\Lmap:\oset\too V$, and a flat $m$-cochain
$X$ in $V$, let $\Lmap^{\#}\left(\cochain\right)$ be the flat $r$-cochain
in $U$ defined by the relation 
\begin{equation}
\Lmap^{\#}\left(\cochain\right)\left(T\right)=\cochain\left(\Lmap_{\#}\left(T\right)\right),\quad\text{for all}\quad T\in F_{r}(U).\label{eq:pullback_flat_cochain}
\end{equation}
The flat $r$-cochain $\Lmap^{\#}(\cochain)$ is represented by the
flat $r$-form $\Lmap^{\#}\left(D_{X}\right)$, the pullback of the
flat $r$-form $D_{X}$ representing $\cochain$ by the map $\Lmap$.
Note that it follows from Rademacher's theorem, \cite[Section 3.1.6]{Federer1969},
that $D\Lmap$ exists $\lusb^{n}$-almost everywhere in $U$. We note
that as Rademacher's applies to a globally Lipschitz map and we consider
locally Lipschitz maps, a generalization of the theorem, Stepanov's
theorem \cite[Theorem 3.4]{Heinonen2000}, may be used to prove the
$\lusb^{n}$-almost existence of $D\Lmap$. This does not limit the
validity of Equation (\ref{eq:pullback_flat_form}), as a flat form
is defined only $\lusb^{n}$-almost everywhere.

Consider a locally Lipschitz map $\Lmap:\oset\too V$ from an open
set $\oset\subset\reals^{n}$ to an open set $V\subset\reals^{m}$.
For a flat $n$-cochain $X$ in $V$ and a current $T_{B}$ induced
by an $L^{n}$-summable set $B$ in $U$, one has 
\begin{equation}
\begin{split}\Lmap^{\#}\left(X\right)(T_{B}) & =\int_{B}\Lmap^{\#}D_{X}dL^{n},\\
 & =\int_{B}D_{X}\left(\Lmap(x)\right)\left(D\Lmap(x)(e_{1})\wedge\dots\wedge D\Lmap(x)(e_{n})\right)d\lusb_{x}^{n},\\
 & =\int_{B}D_{X}\left(\Lmap(x)\right)\left(e_{1}\wedge\dots\wedge e_{n}\right)J_{\Lmap}(x)d\lusb_{x}^{n},\\
 & =\int_{\Lmap\left\{ B\right\} }\sum_{x\in\Lmap^{-1}\left(y\right)}D_{X}\left(y\right)dH_{y}^{n}.
\end{split}
\end{equation}
In the last equation the area formula for Lipschitz maps \cite[Section 2.1.2]{Giaquinta1998}
was applied and $J_{\Lmap}(x)$ is the Jacobian determinant of $\Lmap$
at $x$. In case $\Lmap:U\to V$ is injective with $U\subset\reals^{n}$
and $V\subset\reals^{n}$, we have 
\begin{equation}
\Lmap^{\#}\left(X\right)(T_{B})=X\left(\Lmap_{\#}T_{B}\right)=\int_{\Lmap\left\{ B\right\} }D_{X}\left(y\right)d\lusb_{y}^{n}=X\left(T_{\Lmap\left\{ B\right\} }\right),
\end{equation}
thus, $\Lmap_{\#}T_{B}=T_{\Lmap\left\{ B\right\} }$. In particular,
for a body $\part$, and an injective Lipschitz map $\Lmap$ we note
that 
\begin{equation}
\Lmap_{\#}T_{\part}=T_{\Lmap\left\{ \part\right\} }.\label{eq:conf_T_P=00003DT_conf_P}
\end{equation}
For the material surface $T_{\bnd P}$, Equation (\ref{eq:bndF=00003DFbnd})
gives 
\[
\Lmap_{\#}(T_{\bnd\part})=\Lmap_{\#}(\bnd T_{\part})=\bnd\Lmap_{\#}(T_{\part})=\bnd T_{\Lmap\left\{ \part\right\} },
\]
and for a material surface $T_{\surface}$ Equation (\ref{eq:T_S=00003Dbnd_T_P_res_S})
implies that 
\begin{equation}
\Lmap_{\#}\left(T_{\surface}\right)=T_{\Lmap\left\{ \surface\right\} }.\label{eq:F(T_S)=00003DT_F(S)}
\end{equation}

%% file: sharp.tex
\chapter{The representation of fields over bodies \label{sec:Sharp_functions}}

A real valued field over a body $\part$ will be represented below
by the product of the current $T_{\part}$ and a \emph{sharp function}---a
real valued locally Lipschitz mapping. (The terminology is due to
Whitney \cite[Section V.4]{Whitney1957}.) The space of sharp functions
will be denoted by $\SS{\oset}$. 

A sharp function $\Smap\in\SS{\oset}$ defines a flat $0$-cochain
$\alpha_{\Smap}$ on $\oset$ as follows. Let $\xi$ be an $L^{n}\rest U$-measurable
function compactly supported in $\oset$. Then, $L^{n}\wedge\xi$
is a $0$-current of finite mass in $\oset$ as defined in Equation
(\ref{eq:L^n_0_current}). We set 
\begin{eqnarray}
\alpha_{\Smap}(L^{n}\wedge\xi) & = & \int_{\oset}\phi(x)\left(\xi(x)\right)dL_{x}^{n}.
\end{eqnarray}
For a compactly supported $L^{n}\rest U$ measurable $1$-vector field
$\eta$, $L^{n}\wedge\eta$ is a $1$-current of finite mass in $\oset$
defined in Equation (\ref{eq:L^n_m_current}). Using the existence
of the weak exterior derivative $\wcbd\form$, $\lusb^{n}\rest\oset$-almost
everywhere, we set 
\begin{equation}
\alpha_{\Smap}\left(\bnd\left(L^{n}\wedge\eta\right)\right)=\int_{\oset}\wcbd\phi\left(\eta(x)\right)dL_{x}^{n},
\end{equation}
 and obtain expressions analogous to Wolfe's representation theorem
(Equation (\ref{eq:Wolfe's_representation})). Let $\chain\in F_{0}(\oset)$
be a flat $0$-chain in $\oset$. Applying Theorem \ref{thm:Federer_representation_flat_chains},
$\chain$ may be expressed as $\chain=L^{n}\wedge\xi+\bnd\left(L^{n}\wedge\eta\right)$
with $\xi$ and $\eta$ as defined above. Set
\begin{equation}
\alpha_{\phi}(\chain)=\alpha_{\Smap}\left(L^{n}\wedge\xi+\bnd\left(L^{n}\wedge\eta\right)\right),
\end{equation}
so that $\alpha_{\phi}$ defines a continuous, linear function of
flat $0$-chains. Applying Equation (\ref{eq:F_norm_cochain}) we
obtain 
\begin{equation}
F\left(\alpha_{\Smap}\right)=\sup_{x\in\oset}\left\{ \mass{\Smap(x)},\mass{\wcbd\Smap(x)}\right\} .
\end{equation}

For $\chain\in F_{r}(\oset)$ and $\Smap\in\SS{\oset}$, define the
multiplication $\Smap\chain$ by $\Smap\chain=\alpha_{\Smap}\irest\chain$
using the interior product as defined in Equation (\ref{eq:interior_product_chain}).
That is, 
\begin{equation}
\Smap\chain(\omega)=\left(\alpha_{\Smap}\irest\chain\right)(\omega)=(\alpha_{\Smap}\wedge\omega)(\chain),\quad\text{for all}\quad\omega\in\D^{r}(\oset),\label{eq:sharp_times_flat_chain}
\end{equation}
where $\alpha_{\Smap}\wedge\omega$ is the flat $r$-cochain represented
by the flat $r$-form $\Smap\wedge\omega$. Note that by Equation
(\ref{eq:sharp_times_flat_chain}) 
\begin{equation}
\spt\left(\Smap\chain\right)\subset\spt\left(\Smap\right)\cap\spt\left(\chain\right).\label{eq:spt_sharp_chain}
\end{equation}
For the boundary of $\Smap\chain$ we first note that 
\begin{equation}
\bnd\left(\Smap\chain\right)\left(\omega\right)=\Smap\chain\left(\cbnd\omega\right)=\left(\alpha_{\form}\wedge\cbnd\omega\right)\chain,\quad\text{for all}\quad\omega\in\D^{r-1}(\oset).
\end{equation}
By Equation (\ref{eq:coboundary_of_wedge}) 
\begin{equation}
\cbnd\left(\alpha_{\form}\wedge\omega\right)=\left(\cbnd\alpha_{\form}\right)\wedge\omega+\alpha_{\form}\wedge\cbnd\omega,
\end{equation}
so that
\begin{equation}
\begin{split}\bnd\left(\Smap\chain\right)\left(\omega\right) & =\left(\cbnd\left(\alpha_{\form}\wedge\omega\right)-\left(\cbnd\alpha_{\form}\right)\wedge\omega\right)\chain,\\
 & =\left(\Smap\bnd\chain-\cbnd\alpha_{\form}\irest\chain\right)\left(\omega\right).
\end{split}
\end{equation}
Hence we can write 
\begin{equation}
\bnd\left(\Smap\chain\right)=\Smap\bnd\chain-\cbnd\alpha_{\form}\irest\chain.\label{eq:bnd_sharp_times_chain-1}
\end{equation}

\begin{rem}
The multiplication of sharp functions and chains was originally defined
in \cite[Section VII.1]{Whitney1957} using the notion of continuous
chains which are $r$-vector field approximations of $r$-chains.

\end{rem}

\begin{prop}
\label{prop:sharp_times_normal_and_flat}Given a sharp function $\phi$,
for $\chain\in N_{r,\compact}(\oset)$

\begin{equation}
N_{r,K}\left(\Smap\chain\right)\leq\left(\sup_{x\in\compact}\mass{\Smap(x)}+r\Lip_{\Smap,\compact}\right)N_{r,K}\left(\chain\right),\label{eq:Normal_sharp_norm}
\end{equation}
and for $\chain\in F_{r,\compact}(\oset)$ with $r<n$ (see \cite[p.~208]{Whitney1957})

\begin{equation}
F_{r,K}\left(\Smap\chain\right)\leq\left(\sup_{x\in\compact}\mass{\Smap(x)}+\left(r+1\right)\Lip_{\Smap,K}\right)F_{r,K}\left(\chain\right),\label{eq:Flat_sharp_norm}
\end{equation}
and for $r=n$ 
\begin{equation}
F_{r,K}\left(\Smap\chain\right)\leq\left(\sup_{x\in\compact}\mass{\Smap(x)}\right)F_{r,K}\left(\chain\right).\label{eq:Flat_sharp_n_norm}
\end{equation}

\end{prop}

\begin{proof}
For $\chain\in N_{r,\compact}(\oset)$ we have 
\begin{equation}
\begin{split}\Fmass{\Smap\chain} & =\sup_{\omega\in\D^{r}(\oset)}\frac{\mass{\Smap\chain(\omega)}}{\Fmass{\omega}},\\
 & =\sup_{\omega\in\D^{r}(\oset)}\frac{\mass{\left(\alpha_{\Smap}\wedge\omega\right)\left(\chain\right)}}{\Fmass{\omega}},\\
 & =\sup_{\omega\in\D^{r}(\oset)}\frac{\mass{\int_{\oset}\left(\Smap(x)\omega(x)\right)\left(\vec{T}_{\chain}(x)\right)d\measure{\chain}}}{\Fmass{\omega}},\\
 & \leq\sup_{\omega\in\D^{r}(\oset)}\frac{\sup_{x\in\compact}\|\left(\Smap(x)\omega(x)\right)\|\Fmass{\chain}}{\Fmass{\omega}},\\
 & \le\sup_{x\in\compact}\mass{\Smap(x)}\Fmass{\chain},
\end{split}
\label{eq:Mass_chain_times_sharp_function}
\end{equation}
where in the third line we used the representation by integration
of $\chain$ and in the fourth line the term $\sup_{x\in\compact}\mass{\Smap(x)}$
was extracted since $\spt(\chain)\subset\compact$.

In order to examine the term $\Fmass{\bnd\left(\Smap\chain\right)}$,
we first apply Equation (\ref{eq:bnd_sharp_times_chain-1}) 
\begin{equation}
\Fmass{\bnd\left(\Smap\chain\right)}\leq\Fmass{\form\bnd\chain}+\Fmass{\cbnd\alpha_{\form}\irest\chain}.
\end{equation}
For the first term on the right-hand side we have, 
\begin{equation}
\Fmass{\form\bnd\chain}=\sup_{\omega\in\D^{r-1}(\oset)}\frac{\mass{\alpha_{\Smap}\wedge\omega(\bnd\chain)}}{\Fmass{\omega}}\leq\left(\sup_{x\in\compact}\mass{\Smap(x)}\right)\Fmass{\chain}.
\end{equation}
For the second term,
\begin{equation}
\begin{split}\Fmass{\cbnd\alpha_{\form}\irest\chain} & =\sup_{\omega\in\D^{r-1}(\oset)}\frac{\mass{\int_{\oset}\cbnd\alpha_{\Smap}\wedge\omega\left(\vec{T}_{\chain}\right)d\measure{\chain}}}{\Fmass{\omega}},\\
 & \leq\sup_{\omega\in\D^{r-1}(\oset)}\frac{\sup_{x\in\compact}\|\wcbd\Smap(x)\wedge\omega(x)\|\Fmass{\chain}}{\Fmass{\omega}},\\
 & \leq\sup_{\omega\in\D^{r-1}(\oset)}\left(\begin{array}{c}
r\\
1
\end{array}\right)\frac{\sup_{x\in\compact}\mass{\wcbd\Smap(x)}\Fmass{\omega}\Fmass{\chain}}{\Fmass{\omega}},\\
 & =r\left(\sup_{x\in\compact}\mass{\wcbd\Smap(x)}\right)\Fmass{\chain},
\end{split}
\end{equation}
 where in the third line we used the fact that for an $l$-form $\omega$
and a $k$-form $\omega'$ 
\begin{equation}
\Fmass{\omega\wedge\omega'}\leq\left(\begin{array}{c}
l+k\\
k
\end{array}\right)\Fmass{\omega}\Fmass{\omega'},\label{eq:Mass_for_wedge_forms}
\end{equation}
as is shown in \cite{Federer1960}.

One concludes that 
\begin{equation}
\begin{split}\Fnormal{\Smap\chain} & =\Fmass{\Smap\chain}+\Fmass{\bnd\left(\Smap\chain\right)},\\
 & \leq\sup_{x\in\compact}\mass{\Smap(x)}\Fmass{\chain}+\sup_{x\in\compact}\mass{\Smap(x)}\Fmass{\bnd\chain}+r\Lip_{\Smap,\compact}\Fmass{\chain},\\
 & \leq\left(\sup_{x\in\compact}\mass{\Smap(x)}+r\Lip_{\Smap,\compact}\right)\Fnormal{\chain}.
\end{split}
\label{eq:Normal_chain_times_sharp_function}
\end{equation}

For a flat $r$-chain $\chain\in F_{r,\compact}(\oset)$ we use the
representation given in Equation (\ref{eq:T=00003DR+bndS}) by $\chain=R+\bnd S$
so that $F_{\compact}(\chain)=\Fmass R+\Fmass S.$ We first observe
that 

\begin{equation}
\begin{split}\Fmass{\cbnd\alpha_{\Smap}\irest S} & =\sup_{\omega\in\D^{r}(\oset)}\frac{\cbnd\alpha_{\Smap}\irest S(\omega)}{\Fmass{\omega}},\\
 & =\sup_{\omega\in\D^{r}(\oset),\,\spt(\omega)\subset\compact}\frac{\left(\cbnd\alpha_{\Smap}\wedge\omega\right)(S)}{\Fmass{\omega}},\\
 & \leq\frac{\Fmass S\Fmass{\cbnd\alpha_{\Smap}\wedge\omega}}{\Fmass{\omega}},\\
 & \leq\frac{\Fmass S}{\Fmass{\omega}}\left(\begin{array}{c}
r+1\\
r
\end{array}\right)\Fmass{\omega}\sup_{x\in\compact}\mass{\wcbd\Smap(x)},
\end{split}
\end{equation}
and conclude that 
\begin{equation}
\Fmass{\cbnd\alpha_{\Smap}\irest S}\leq\left(r+1\right)\Lip_{\Smap,\compact}\Fmass S.\label{eq:inequality_a}
\end{equation}

Estimating $F_{K}(\phi\chain)$, one has

\begin{equation}
\begin{split}F_{\compact}(\Smap\chain) & =F_{\compact}\left(\Smap R+\Smap\bnd S\right),\\
 & \leq F_{\compact}\left(\Smap R\right)+F_{\compact}\left(\Smap\bnd S\right),\\
 & \leq F_{\compact}\left(\Smap R\right)+F_{\compact}\left(\cbnd\alpha_{\Smap}\irest S+\bnd(\Smap S)\right),\\
 & \leq F_{\compact}\left(\Smap R\right)+F_{\compact}\left(\cbnd\alpha_{\Smap}\irest S\right)+F_{\compact}\left(\bnd(\Smap S)\right),\\
 & \leq F_{\compact}\left(\Smap R\right)+F_{\compact}\left(\cbnd\alpha_{\Smap}\irest S\right)+F_{\compact}\left((\Smap S)\right),\\
 & \leq\Fmass{\Smap R}+\Fmass{\cbnd\alpha_{\Smap}\irest S}+\Fmass{\Smap S},\\
 & \leq\sup_{x\in\compact}\mass{\Smap(x)}\Fmass R+(r+1)\Lip_{\Smap,\compact}\Fmass S+\sup_{x\in\compact}\mass{\Smap(x)}\Fmass S,\\
 & \leq\left\{ \sup_{x\in\compact}\mass{\Smap(x)}+(r+1)\Lip_{\Smap,\compact}\right\} \left(\Fmass R+\Fmass S\right),\\
 & =\left\{ \sup_{x\in\compact}\mass{\Smap(x)}+(r+1)\Lip_{\Smap,\compact}\right\} \Fflat{\chain},
\end{split}
\label{eq:Flat_chain_times_sharp_function}
\end{equation}
where in the third line we used Equation (\ref{eq:bnd_sharp_times_chain-1}),
in the sixth line we used Equation (\ref{eq:F<M}), and in the seventh
line we used Equation (\ref{eq:inequality_a}). For the case $r=n$
Equation \ref{eq:Flat_sharp_n_norm} follows from the fact that $S=0$.
\end{proof}
Note that by Equation (\ref{eq:Normal_sharp_norm}) it follows that
\begin{equation}
\Fnormal{\Smap\chain}\leq(r+1)\norm{\Smap}{\Lip,\compact}\Fnormal{\chain},
\end{equation}
and by Equation (\ref{eq:Flat_sharp_norm}) if follows that for $r<n$
\begin{equation}
F_{\compact}(\Smap\chain)\leq(r+2)\norm{\Smap}{\Lip,\compact}F_{\compact}(\chain).
\end{equation}

The vector space of sharp functions defined on $\oset$ and valued
in $\reals^{m}$ is identified as the space of $m$-tuples of real
valued sharp functions defined on $\oset$ i.e. $\SS{\oset,\reals^{m}}=\left[\SS{\oset}\right]^{m}$.
For $\Smap\in\SS{\oset,\reals^{m}}$ and $\chain\in F_{r,\compact}(\oset)$
the flat $r$-chain $\Smap\chain$ is viewed as an element of the
vector space of $\left(F_{r,\compact}(\oset)\right)^{m}$, \textit{i.e.},
an $m$-tuple of flat $r$-chains in $\oset$ with $\left(\Smap\chain\right)_{i}=\Smap_{i}\chain$.

%% file: configuration.tex
\chapter{Configuration space and virtual velocities \label{sec:configuration-space-and}}

Traditionally, a configuration of a body $\part$ is viewed as a mapping
$\part\to\reals^{n}$ which preserves the basic properties assigned
to bodies and material surfaces. Guided by our initial definition
of a body $T_{\part}$ as a current induced by $\part$, a set of
finite perimeter in the open set $\body$, a configuration of the
body $\part$ is defined as a mapping $\conf_{\part}\in\Limb(\part,\reals^{n})$.
To distinguish it from a configuration of the universal body to be
considered below, such an element, $\conf_{\part}$, will be referred
to as a \textit{local configuration}. The choice of Lipschitz type
configurations is a generalization of the traditional choice of $C^{1}$-embeddings
usually taken in continuum mechanics.

It is natural therefore to refer to $\confs_{\part}=\Limb(\part,\reals^{n})$
as the \emph{configuration space of the body }$\part$. Since a body
is a compact set, it follows from Theorem \ref{thm:Lipschitz_embedding_open_set}
that $\confs_{\part}$ is an open subset of the Banach space $\Lip(\part,\reals^{n})\cong\Lip\left(\conf_{\part}\left\{ \part\right\} ,\reals^{n}\right)$. 

For $\part,\part'\in\Omega_{\body}$ the local configurations $\conf_{\part},\conf_{\part'}$
are said to be compatible if 
\begin{equation}
\conf_{\part}\mid_{\part\cap\part'}=\conf_{\part'}\mid_{\part\cap\part'}.
\end{equation}
 Note that the intersection of two sets of finite perimeter is a set
of finite perimeter, thus, the restricted map may be viewed as the
configuration of the body $\part\cap\part'$. 

A \textit{system of compatible configurations} $\gconf$, is a collection
of compatible local configurations $\conf=\left\{ \conf_{\part}\mid\part\in\Omega_{\body}\right\} $.
Clearly, a system of compatible configuration is represented by a
unique element of $\Limb\left(\body,\reals^{n}\right)$. An element
$\gconf\in\Limb\left(\body,\reals^{n}\right)$ will be referred to
as a \textit{global configuration,} and the \textit{global configuration
space} \textit{$\confs$} is the collection of all global configurations,
\textit{i.e.,}
\begin{equation}
\confs=\Limb\left(\body,\reals^{n}\right).
\end{equation}
We will view the configuration space as a trivial infinite dimensional
differentiable manifold, specifically, a trivial manifold modeled
on a locally convex topological vector space as in \cite[Chapter 9]{Michor1980}.

It is noted, in particular, that a Lipschitz embedding is injective
and the image of a set of a finite perimeter in $\body$ is a set
of finite perimeter in $\reals^{n}$. In addition, as Chapter \ref{sec:On-Lipschitz-mappings}
indicates, Lipschitz mappings are the natural morphism in the category
of sets of finite perimeters and in the category of flat chains. Thus,
an element $\conf\in\confs$ preserves the structure of bodies and
material surfaces as required. That is, every $\conf\in\confs$ induces
a map $\conf_{\#}$ of flat chains. For any $T_{\part}\in\Omega_{\body}$,
the current $\conf_{\#}\left(T_{\part}\right)$ is an element of $N_{n}\left(\reals^{n}\right)$,
and for any $T_{\surface}\in\bnd\Omega_{\body}$, the current $\conf_{\#}\left(T_{\surface}\right)$
is an $(n-1)$-chain of finite mass in $\reals^{n}$. By Equations
(\ref{eq:conf_T_P=00003DT_conf_P}) and (\ref{eq:F(T_S)=00003DT_F(S)})
it follows that $\conf_{\#}\left(T_{\part}\right)=T_{\conf\left\{ \part\right\} }$
and $\conf_{\#}\left(T_{\surface}\right)=T_{\conf\left\{ \surface\right\} }$.
Applying Equation (\ref{eq:Lipschitz_normal_norm}), one obtains for
every $T_{\surface}\in\bnd\Omega_{\body}$ that 
\begin{equation}
\Fmass{\conf_{\#}\left(T_{\surface}\right)}\leq\Fmass{T_{\surface}}\left(\Lip_{\conf,\hat{\surface}}\right)^{n-1}.
\end{equation}
By Equation (\ref{eq:Lipschitz_mass_norm}), for every $T_{\part}\in\Omega_{\body}$,
\begin{eqnarray}
\Fnormal{\conf_{\#}\left(T_{\part}\right)} & \leq & N\left(T_{\part}\right)\sup\left\{ \left(\Lip_{\conf,\part}\right)^{n},\left(\Lip_{\conf,\part}\right)^{n-1}\right\} .
\end{eqnarray}

For a global configuration $\conf$, let $\conf\left(\Omega_{\body}\right)$
denote the collection of images of bodies under the configuration
$\conf$, \textit{i.e.},
\begin{equation}
\conf\left(\Omega_{\body}\right)=\left\{ \conf_{\#}\left(T_{\part}\right)\mid T_{\part}\in\Omega_{\body}\right\} .
\end{equation}
Similarly, the collection of surfaces at the configuration $\conf$
is 
\begin{equation}
\conf\left(\bnd\Omega_{\body}\right)=\left\{ \conf_{\#}\left(T_{\surface}\right)\mid T_{\surface}\in\bnd\Omega_{\body}\right\} .
\end{equation}

A\textit{ global virtual velocity at the configuration} $\conf$ is
identified with an element of the tangent space to $\confs$ at $\conf$.
By Theorem \ref{thm:Lipschitz_embedding_open_set}, $\Lip(\body,\reals^{n})$
is naturally isomorphic to any tangent space to $\confs$. Moreover,
$\conf$ induces an isomorphism $\Lip(\body,\reals^{n})\cong\Lip\left(\conf\left\{ \body\right\} ,\reals^{n}\right)$
and \textit{an Eulerian virtual velocity} is viewed as an element
of $\Lip\left(\conf\left\{ \body\right\} ,\reals^{n}\right)$. In
what follows, we refer to $\Lip\left(\conf\left\{ \body\right\} ,\reals^{n}\right)$
as the space of global virtual velocities at the configuration $\conf$
and use the abbreviated notation $\virvs$ for it. Naturally, an element
of $\virvs$ may be identified with an $n$-tuple of sharp functions
defined on $\conf\left\{ \body\right\} $, \textit{i.e.}, using the
Whitney topology on $\Lip(\conf\left\{ \body\right\} )$, $\virvs=\left[\Lip(\conf\left\{ \body\right\} )\right]^{n}$.

Focusing our attention to a particular body $\part$, one may make
use of the approach of \cite{Segev1986} and define a \emph{virtual
velocity of a body $\part$ at a configuration} $\conf_{\part}\in\confs_{\part}$
as an element $\virv_{\part}$ in the tangent space $T_{\conf_{\part}}\confs_{\part}$.
It follows from Theorem \ref{thm:Lipschitz_embedding_open_set} that
one may make the identifications $T_{\conf_{\part}}\confs_{\part}\cong\Lip(\part,\reals^{n})\cong\Lip\left(\conf_{\part}\left\{ \part\right\} ,\reals^{n}\right)$. 
\begin{thm}
\label{thm:RestrOfConf} For every body $\part$, and every $\conf_{\part}\in\confs_{\part},$
and every $\conf\in\confs$ such that $\conf\mid_{\part}=\conf_{\part},$
the restriction mapping 
\begin{equation}
\rho_{\part}:T_{\conf}\confs\longrightarrow T_{\conf_{\part}}\confs_{\part}
\end{equation}
is surjective.\end{thm}
\begin{proof}
We recall that Kirszbraun's theorem asserts that a Lipschitz mapping
$f:A\to\reals^{m}$ defined on a set $A\subset\reals^{n}$ may be
extended to to a Lipschitz function $F:\reals^{n}\to\reals^{m}$ having
the same Lipschitz constant (see \cite[Section 2.10.43]{Federer1969}
or \cite[Section 6.2]{Heinonen2000}). It follows immediately that
any $\virv_{\part}\in\Lip(\part,\reals^{n})$ may be extended to an
element $\virv\in\Lip(\body,\reals^{n})$.
\end{proof}
Anticipating the properties of systems of forces to be considered
below, we wish to provide the collection of restrictions of global
virtual velocities to the various bodies with a finer structure than
that provided by the $\norm{\cdot}{\Lip,\compact}$-semi-norms. In
particular, when considering the restriction $\virv\mid_{\part}$
of a global virtual velocity $\virv$ to a body $\part$, we wish
that the magnitude of the resulting object will reflect the mass of
$\part$. The \textit{local virtual velocity for the body $T_{\part}$
at the configuration} $\conf$ \textit{induced by the global virtual
velocity} $\virv\in\virvs$ is defined as the $n$-tuple of normal
$n$-currents given by the products $\virv\conf_{\#}\left(T_{\part}\right)$
such that 
\begin{equation}
\left[\virv\conf_{\#}\left(T_{\part}\right)\right]_{i}=\virv_{i}\conf_{\#}\left(T_{\part}\right),\quad\text{for all}\quad i=1,\dots,n.\label{eq:n_tuple_of_currents}
\end{equation}
By Equations (\ref{eq:Mass_chain_times_sharp_function}) and (\ref{eq:Normal_sharp_norm}),
each component $\left[\virv\conf_{\#}\left(T_{\part}\right)\right]_{i}$
is a normal $n$-current such that
\begin{equation}
\begin{split}\Fmass{\left[\virv\conf_{\#}\left(T_{\part}\right)\right]_{i}} & \le\sup_{y\in\conf\left\{ \part\right\} }\mass{\virv_{i}(y)}\Fmass{\conf_{\#}\left(T_{\part}\right)},\\
 & \le\sup_{y\in\conf\left\{ \part\right\} }\mass{\virv_{i}(y)}\left(\Lip_{\conf,\part}\right)^{n}\Fmass{T_{\part}},
\end{split}
\label{eq:BoundOnMass_vTp}
\end{equation}
and 
\begin{equation}
\begin{split}N\left(\left[\virv\conf_{\#}\left(T_{\part}\right)\right]_{i}\right) & \leq\left(\left(\sup_{y\in\conf\left\{ \part\right\} }\mass{\virv_{i}(y)}\right)+n\Lip_{\virv_{i},\conf\left\{ \part\right\} }\right)N\left(\conf_{\#}\left(T_{\part}\right)\right),\\
 & \le\left(\left(\sup_{y\in\conf\left\{ \part\right\} }\mass{\virv_{i}(y)}\right)+n\Lip_{\virv_{i},\conf\left\{ \part\right\} }\right)\\
 & \qquad\qquad\times\sup\left\{ \left(\Lip_{\conf,\part}\right)^{n},\left(\Lip_{\conf,\part}\right)^{n-1}\right\} N(T_{\part}).
\end{split}
\end{equation}
In other words, the mapping $\virvs\times\Omega_{\body}\to\D_{m}\left(\body\right)$
given by $(v,T_{\part})\mapsto\virv\conf_{\#}\left(T_{\part}\right)$
is continuous with respect to both the mass norm and the normal norm. 

Similarly, the assignment of a virtual velocity $\virv\in\virvs$
to a material surface $T_{\surface}$ induces an $n$-tuple of $(n-1)$-chains
defined by the multiplication $\virv\conf_{\#}\left(T_{\surface}\right)$.
Each component $\left[\virv\conf_{\#}\left(T_{\surface}\right)\right]_{i}$
is a chain of finite mass and applying Equation (\ref{eq:Mass_chain_times_sharp_function}),
one obtains 
\begin{equation}
\begin{split}M\left(\left[\virv\conf_{\#}\left(T_{\surface}\right)\right]_{i}\right) & \leq\left(\sup_{y\in\conf\left\{ \hat{\surface}\right\} }\mass{\virv_{i}(y)}\right)M\left(\conf_{\#}\left(T_{\surface}\right)\right),\\
 & \leq\left(\sup_{y\in\conf\left\{ \hat{\surface}\right\} }\mass{\virv_{i}(y)}\right)\left(\Lip_{\conf,\hat{\surface}}\right)^{n-1}M\left(T_{\surface}\right).
\end{split}
\label{eq:BoundOnMass_vTs}
\end{equation}

%% file: Transport.tex
\chapter{Density transport theorem \label{chap:Reynolds-transport-theorem}}

In this chapter we apply the general setting presented thus far and
present a density transport theorem which is analogous to Reynolds
transport theorem for an implicit time dependent property. Using the
framework introduced in Chapter \ref{sec:configuration-space-and},
a \textit{motion} is defined as a mapping 
\begin{equation}
\motion:\reals\times\body\to\reals^{n},
\end{equation}
 such that for every $t\in\reals$ the map $\conf_{t}:\body\to\reals^{n}$
defined by 
\begin{equation}
\conf_{t}(x)=\motion(t,x),\quad\text{for all}\quad x\in\body,
\end{equation}
is a global configuration as presented in Chapter \ref{sec:configuration-space-and},
\textit{i.e.}, a Lipschitz embedding $\conf_{t}\in\Limb\left(\body,\reals^{n}\right)$.

In the spirit of Chapter \ref{sec:Sharp_functions}, a general Lagrangian
representation of an intensive property is assumed to be given by
\begin{equation}
\psi:\body\to\reals,
\end{equation}
where we assume that $\psi$ is a sharp function \textit{i.e.,} a
real valued, bounded, locally Lipschitz function. The extensive property
associated with $\psi$ and the body $T_{\part}$ is defined as the
multiplication $\psi T_{\part}$, which, by Proposition \ref{prop:sharp_times_normal_and_flat},
is a normal $n$-current in $\body$. For any $\omega\in\D^{n}\left(\reals^{n}\right)$
\begin{equation}
\begin{split}\conf_{t\#}\left(\psi T_{\part}\right)(\omega) & =\int_{\part}\psi(x)\omega\left(\conf_{t}(x)\right)\left(\bigwedge_{n}D\conf_{t}\left(e_{1}\wedge\dots\wedge e_{n}\right)\right)d\lusb_{x}^{n},\\
 & =\int_{\part}\psi\left(\conf_{t}^{-1}\left(\conf_{t}\left(x\right)\right)\right)\omega\left(\conf_{t}(x)\right)\left(\bigwedge_{n}D\conf_{t}\left(e_{1}\wedge\dots\wedge e_{n}\right)\right)d\lusb_{x}^{n},\\
 & =\conf_{t\#}\left(T_{\part}\right)\left(\psi_{\conf_{t}}\wedge\omega\right),\\
 & =\psi_{\conf_{t}}\conf_{t\#}\left(T_{\part}\right)\left(\omega\right).
\end{split}
\end{equation}
Thus, 
\begin{equation}
\conf_{t\#}\left(\psi T_{\part}\right)=\psi_{\conf_{t}}\conf_{t\#}\left(T_{\part}\right),
\end{equation}
where 
\begin{equation}
\psi_{\conf_{t}}=\psi\circ\conf_{t}^{-1}:\conf_{t}\left(\body\right)\to\reals,
\end{equation}
is viewed as the Eulerian representation of the property $\psi$. 

In order to develop a density transport theorem we wish to investigate
the term 
\begin{equation}
\frac{d}{dt}\left(\conf_{t\#}\left(\psi T_{\part}\right)\right)_{t=0},
\end{equation}
which will be done by applying the homotopy theory for currents and
the formal definition of the derivative such that 
\begin{equation}
\frac{d}{dt}\left(\conf_{t\#}\left(\psi T_{\part}\right)\right)_{t=0}=\lim_{\varepsilon\to0}\left[\frac{\conf_{\varepsilon\#}\left(\psi T_{\part}\right)-\conf_{0\#}\left(\psi T_{\part}\right)}{\varepsilon}\right],\label{eq:dot_of_current}
\end{equation}
we first recall some basic properties of the homotopy theorem for
currents. 

Let $\oset\subset\reals^{n}$ and $V\subset\reals^{m}$ be open sets
with $T\in\D_{k}(\oset)$ and $S\in\D_{l}(V)$. Then, the Cartesian
product of $T$ and $S$ is an element of $\D_{k+l}\left(\oset\times V\right)$
denoted by $T\times S$ and defined as follows. Let $p,\: q$ be the
projection mappings 
\begin{equation}
p:\oset\times V\to U,\quad q:\oset\times V\to V.
\end{equation}
For $\alpha\in\D^{r}(\oset)$ and $\beta\in\D^{m+k-r}(V)$, we note
that $p^{\#}(\alpha)\wedge q^{\#}(\beta)$ is an element of $\D^{m+k}(\oset\times V)$.
Thus 
\begin{equation}
\left(T\times S\right)\left(p^{\#}(\alpha)\wedge q^{\#}(\beta)\right)=\begin{cases}
T(\alpha)S(\beta), & \text{in case }r=k,\\
0 & \text{in case }r\not=k.
\end{cases}\label{eq:currents_product_def}
\end{equation}
For the properties of the Cartesian products of currents we refer
to \cite[Section 4.1.8]{Federer1969}.

Let $\oset\subset\reals^{n}$ be an open set and let $f$ and $g$
be locally Lipschitz mappings of $\oset$ into $\reals^{m}$. For
an open set $A$ of $\reals$ such that $\left[0,1\right]\subset A$,
a Lipschitz homotopy from $f$ to $g$ is a map 
\begin{equation}
h:A\times\oset\to\reals^{m},
\end{equation}
such that 
\begin{equation}
h\left(0,x\right)=f(x),\;\text{and }h(1,x)=g(x),
\end{equation}
for all $x\in\oset$. A Lipschitz homotopy $h$ is said to be a linear
homotopy if
\begin{equation}
h(\tau,x)=(1-\tau)f(x)+\tau g(x).
\end{equation}
In the following, we will use the following notation 
\begin{equation}
h_{\tau}(x)=h(\tau,x),\;\text{for all}\, x\in\oset,
\end{equation}
and 
\[
\dot{h}_{\tau}:\oset\to\reals^{m},\qquad\dot{h}_{\tau}(x)=Dh(\tau,x)\left(1,0\right),\;\text{for all }x\in\oset,
\]
where in the preceding equation $0$ is the zero element in $\reals^{n}$.
For $T\in\D_{r}(\oset)$ and a homotopy $h$ between $f$ and $g$,
the \textit{$h$ deformation chain of $T$} is defined as the current
\begin{equation}
h_{\#}\left(\left[0,1\right]\times T\right)\in\D_{r+1}\left(\reals^{m}\right).
\end{equation}
The properties of the $h$ deformation chain are further investigated
in \cite[Section 4.1.9]{Federer1969} where it shown that for $r>0$
\begin{equation}
g_{\#}\left(T\right)-f_{\#}\left(T\right)=\bnd h_{\#}\left(\left[0,1\right]\times T\right)+h_{\#}\left(\left[0,1\right]\times\bnd T\right).\label{eq:homotopy_formula_for_currents}
\end{equation}
For an $r$-current $T$ which is represented by integration and $\omega\in\D^{r+1}\left(\reals^{m}\right)$,
\begin{equation}
\begin{split}h_{\#}\left(\left[0,1\right]\times T\right)(\omega) & =\int_{[0,1]}\left[\int_{\oset}\omega\left(h_{\tau}(x)\right)\left(\dot{h}_{\tau}(x)\wedge\left(Dh_{\tau}(x)\vec{T}(x)\right)\right)d\measure T\right]d\lusb_{\tau}^{1}\end{split}
.\label{eq:Integral_representation_of_deformation_chain}
\end{equation}

We now return to Equation (\ref{eq:dot_of_current}) and let $h^{\varepsilon}:[0,1]\times\body\to\reals^{n}$
be the \emph{linear homotopy} between $\conf_{0}$ and $\conf_{\varepsilon}$
\textit{i.e.,} $h^{\varepsilon}(0,x)=\conf_{0}(x)$ and $h^{\varepsilon}(1,x)=\conf_{\varepsilon}(x)$
such that 
\begin{equation}
h^{\varepsilon}(\tau,x)=\conf_{0}(x)(1-\tau)+\conf_{\varepsilon}(x)\tau.\label{eq:linear_homotopy}
\end{equation}
Applying the homotopy formula, Equation (\ref{eq:homotopy_formula_for_currents}),
it follows that 
\begin{equation}
\conf_{\varepsilon\#}\left(\psi T_{\part}\right)-\conf_{0\#}\left(\psi T_{\part}\right)=\bnd h_{\#}^{\varepsilon}\left([0,1]\times\psi T_{\part}\right)+h_{\#}^{\varepsilon}\left([0,1]\times\bnd\left(\psi T_{\part}\right)\right).\label{eq:homotopy_currents}
\end{equation}
The following results are independent of any particular homotopy chosen
and a linear homotopy was selected for convenience. Since $h_{\#}^{\varepsilon}\left([0,1]\times\psi T_{\part}\right)$
is an $\left(n+1\right)$-current in $\reals^{n}$, the first term
on the right-hand side of Equation (\ref{eq:homotopy_currents}) vanishes.
Applying Equation (\ref{eq:bnd_sharp_times_chain-1}), we obtain 
\begin{equation}
\begin{split}\conf_{\varepsilon\#}\left(\psi T_{\part}\right)-\conf_{0\#}\left(\psi T_{\part}\right) & =h_{\#}^{\varepsilon}\left([0,1]\times\wcbd\psi\irest T_{\part}\right)+h_{\#}^{\varepsilon}\left([0,1]\times\psi\bnd T_{\part}\right).\end{split}
\end{equation}
Thus, 
\begin{equation}
\frac{d}{dt}\left(\conf_{t\#}\left(\psi T_{\part}\right)\right)_{t=0}=\lim_{\varepsilon\to0}\left[\frac{h_{\#}^{\varepsilon}\left([0,1]\times\wcbd\psi\irest T_{\part}\right)+h_{\#}^{\varepsilon}\left([0,1]\times\psi\bnd T_{\part}\right)}{\varepsilon}\right].
\end{equation}
Each of the terms is examined separately by applying the integral
representation of the $h_{\varepsilon}$ deformation chain as given
by Equation (\ref{eq:Integral_representation_of_deformation_chain}).
As $h^{\varepsilon}$ is a linear homotopy, by direct calculations
\begin{eqnarray}
Dh_{\tau}^{\varepsilon}(x) & = & (1-\tau)D\conf_{0}(x)+\tau D\conf_{\varepsilon}(x),\\
\dot{h}_{\tau}^{\varepsilon}(x) & = & \conf_{\epsilon}(x)-\conf_{0}(x).
\end{eqnarray}
For $\omega\in\D^{n}\left(\reals^{n}\right)$ observe that 
\begin{multline}
\lim_{\varepsilon\to0}\left[\frac{h_{\#}^{\varepsilon}\left([0,1]\times\psi\bnd T_{\part}\right)(\omega)}{\varepsilon}\right]\\
\begin{split} & =\lim_{\varepsilon\to0}\left[\frac{\int_{[0,1]}\left[\int_{\mtb\left(\part\right)}\psi\left(x\right)\omega\left(h_{\tau}^{\varepsilon}\left(x\right)\right)\left(\dot{h}_{\tau}^{\varepsilon}(x)\wedge\left[\bigwedge_{n-1}Dh_{\tau}^{\varepsilon}(x)\right]\vec{\bnd T_{\part}}(x)\right)dH_{x}^{n-1}\right]d\lusb_{\tau}^{1}}{\varepsilon}\right],\\
 & =\int_{[0,1]}\left[\int_{\mtb\left(\part\right)}\left[\psi\left(x\right)\omega\left(\conf_{0}\left(x\right)\right)\left(\virv(x)\wedge\left[\bigwedge_{n-1}D\conf_{0}(x)\right]\vec{\bnd T_{\part}}(x)\right)\right]dH_{x}^{n-1}\right]d\lusb_{\tau}^{1}.
\end{split}
\end{multline}
 Here, $\virv(x)$, defined as 
\begin{equation}
\virv(x)=\lim_{\varepsilon\to0}\frac{\dot{h}_{\tau}^{\varepsilon}(x)}{\varepsilon}=\lim_{\varepsilon\to0}\frac{\conf_{\epsilon}(x)-\conf_{0}(x)}{\varepsilon},
\end{equation}
is viewed as the velocity of the material point $x\in\body$ at the
time $t=0$. In addition, set $u(x)=D\conf_{0}(x)^{-1}\left(\virv(x)\right)$.
Note that the integrand is independent of $t$, thus
\begin{multline}
\lim_{\varepsilon\to0}\left[\frac{h_{\#}^{\varepsilon}\left([0,1]\times\psi\bnd T_{\part}\right)(\omega)}{\varepsilon}\right]\\
\begin{split} & =\int_{\mtb\left(\part\right)}\left[\psi\left(x\right)\left(\omega\left(\conf_{0}\left(x\right)\right)\right)\left(\left[\bigwedge_{n}D\conf_{0}(x)\right]u(x)\wedge\vec{\bnd T_{\part}}(x)\right)\right]dH_{x}^{n-1},\\
 & =\left(\conf_{0}^{\#}(\omega)\right)\left(\psi u\wedge\bnd T_{\part}\right),\\
 & =\conf_{0\#}\left(\psi u\wedge\bnd T_{\part}\right)(\omega),\\
 & =\psi_{\conf_{0}}\conf_{0\#}\left(u\wedge\bnd T_{\part}\right)(\omega),\\
 & =\psi_{\conf_{0}}v\wedge\conf_{0\#}\left(\bnd T_{\part}\right)(\omega)
\end{split}
\end{multline}
 Here, $u\wedge\bnd T_{\part}$ is defined as the $n$-current such
that  
\begin{equation}
u\wedge\bnd T_{\part}(\omega)=\left(u\irest\omega\right)\left(\bnd T_{\part}\right).
\end{equation}
We use $\bigwedge_{m}D\conf_{0}(x)$ for the map 
\begin{equation}
\bigwedge_{m}D\conf_{0}(x):\bigwedge_{m}\reals^{n}\to\bigwedge_{m}\reals^{n},
\end{equation}
defined by 
\begin{equation}
\bigwedge_{m}D\conf_{0}(x)\left(v_{1}\wedge\dots\wedge v_{m}\right)=\left(D\conf_{0}(x)\left(v_{1}\right)\right)\wedge\dots\wedge\left(D\conf_{0}(x)\left(v_{m}\right)\right),\quad\text{for all}\; v_{1},\dots,v_{m}\in\reals^{n}.
\end{equation}

For the term $\lim_{\varepsilon\to0}\left[\frac{h_{\#}^{\varepsilon}\left([0,1]\times\wcbd\psi\irest T_{\part}\right)}{\varepsilon}\right]$,
we obtain
\begin{multline}
\lim_{\varepsilon\to0}\left[\frac{h_{\#}^{\varepsilon}\left([0,1]\times\wcbd\psi\irest T_{\part}\right)(\omega)}{\varepsilon}\right]\\
\begin{split} & =\lim_{\varepsilon\to0}\left[\frac{\int_{[0,1]}\left[\int_{\part}\omega\left(h_{\tau}^{\varepsilon}\left(x\right)\right)\left(\dot{h}_{\tau}^{\varepsilon}(x)\wedge\left[\bigwedge_{n-1}Dh_{\tau}^{\varepsilon}(x)\right]\wcbd\psi\irest(e_{1}\wedge\dots\wedge e_{n})\right)d\lusb_{x}^{n}\right]d\lusb_{\tau}^{1}}{\varepsilon}\right],\\
 & =1\int_{\part}\left[\omega\left(\conf_{0}\left(x\right)\right)\left(\virv(x)\wedge\left[\bigwedge_{n-1}D\conf_{0}(x)\right]\wcbd\psi\irest(e_{1}\wedge\dots\wedge e_{n})\right)\right]d\lusb_{x}^{n},\\
 & =\int_{\part}\left[\omega\left(\conf_{0}\left(x\right)\right)\left(\left[\bigwedge_{n}D\conf_{0}(x)\right]\wcbd\psi\left(D\conf_{0}^{-1}(x)v(x)\right)(e_{1}\wedge\dots\wedge e_{n})\right)\right]d\lusb_{x}^{n}.
\end{split}
\end{multline}
The term $\wcbd\psi\left(D\conf_{0}^{-1}(x)v(x)\right)$ is identified
with the time-derivative of the Eulerian field describing the property
$\psi$ 
\begin{equation}
\wcbd\psi\left(D\conf_{0}^{-1}(x)v(x)\right)=\frac{d\psi_{\conf_{t}}}{dt}|_{t=0}.
\end{equation}
As a result, 
\begin{equation}
h_{\#}^{\varepsilon}\left([0,1]\times\cbnd\psi\irest T_{\part}\right)(\omega)=\frac{d\psi_{\conf_{t}}}{dt}|_{t=0}\conf_{0\#}\left(T_{\part}\right)(\omega).
\end{equation}
We concluded that
\begin{equation}
\frac{d}{dt}\left(\conf_{t\#}\left(\psi T_{\part}\right)\right)_{t=0}=\psi_{\conf_{0}}v\wedge\conf_{0\#}\left(\bnd T_{\part}\right)+\frac{d\psi_{\conf_{t}}}{dt}|_{t=0}\conf_{0\#}\left(T_{\part}\right),
\end{equation}
where the first term is associated as the flux of the property $\psi$
through the boundary of the body, and the second term is the time
derivative of the property in the domain of the body $T_{\part}$.

%% file: Cauchy_flux.tex
\chapter{Cauchy fluxes\label{sec:Cauchy-fluxes}}

Alluding to the approach of \cite{Segev1986} again, a \emph{force
on a body $\part$ at the configuration $\conf_{\part}\in\confs_{\part}$
}is an element in the dual to the tangent space, $T_{\conf_{\part}}^{*}\confs_{\part}$.
In other words, forces on $\part$ are elements of the infinite dimensional
cotangent bundle $T^{*}\confs_{\part}$. For $\force_{\part}\in T_{\conf_{\part}}^{*}\confs_{\part}$,
and $\virv_{\part}\in T_{\conf_{\part}}\confs_{\part}$, the action
$\force_{\part}(\virv_{\part})$ is interpreted as the virtual power
performed by the force $\force_{\part}$ for the virtual velocity
$\virv_{\part}$. It follows immediately that a force on a body $\part$
at $\conf_{\part}$ may be identified with a linear continuous functional
on the space of Lipschitz mappings. Such functionals are quite irregular
and will not be considered here.

Instead, we use in this chapter the notion of a Cauchy flux at the
configuration $\conf$, as a real valued function operating on the
Cartesian product $\conf\left(\bnd\Omega_{\body}\right)\times W_{\conf}$.
These impose stricter conditions on the force system and resulting
stress fields. The conditions to be imposed still imply that for a
fixed body, a force is a continuous linear functional of the virtual
velocities of that body.

A Cauchy flux represents a system of surface forces operating on the
material surfaces, or more precisely, their images under $\conf$.
For a given surface and a given virtual velocity field, the value
returned by the Cauchy flux mapping is interpreted as the virtual
power (or virtual work) performed by the force acting on the image
of the material surface under $\conf$ for the given virtual velocity.
\begin{defn}
A \textit{Cauchy flux} at the configuration $\conf$ is a mapping
of the form 
\begin{equation}
\ssys_{\conf}:\conf\left(\bnd\Omega_{\body}\right)\times W_{\conf}\to\reals,\label{eq:Cauchy_flux_def}
\end{equation}
 such that the following hold.\end{defn}
\begin{description}
\item [{Additivity}] $\ssys_{\conf}\left(\cdot,\virv\right)$ is additive
for disjoint compatible material surfaces, \emph{i.e.}, for every
$\conf_{\#}\left(T_{\surface}\right),\conf_{\#}\left(T_{\surface'}\right)\in\conf\left(\bnd\Omega_{\body}\right)$
compatible and disjoint, 
\begin{equation}
\ssys_{\conf}\left(\conf_{\#}\left(T_{\surface\cup\surface'}\right),\virv\right)=\ssys_{\conf}\left(\conf_{\#}\left(T_{\surface}\right),\virv\right)+\ssys_{\conf}\left(\conf_{\#}\left(T_{\surface'}\right),\virv\right),\label{eq:Cauchy_flux_additivity}
\end{equation}
holds for every $\virv\in\virvs$.
\item [{Linearity}] $\ssys_{\conf}\left(\conf_{\#}\left(T_{\surface}\right),\cdot\right)$
is a linear function on $\virvs$, \emph{i.e.}, for all $\alpha,\,\beta\in\reals$
and $\virv,\,\virv'\in\virvs$, 
\begin{equation}
\ssys_{\conf}\left(\conf_{\#}\left(T_{\surface}\right),\alpha\virv+\beta\virv'\right)=\alpha\ssys_{\conf}\left(\conf_{\#}\left(T_{\surface}\right),\virv\right)+\beta\ssys_{\conf}\left(\conf_{\#}\left(T_{\surface}\right),\virv'\right)\label{eq:Cauchy_flux_linearity}
\end{equation}
holds for every $\conf_{\#}\left(T_{\surface}\right)\in\conf\left(\bnd\Omega_{\body}\right)$.
\end{description}
 Let $\virv\in\virvs$ and $\conf_{\#}\left(T_{\surface}\right)\in\conf\left(\bnd\Omega_{\body}\right)$,
then, by the linearity of the Cauchy flux, 
\begin{equation}
\Phi_{\conf}\left(\conf_{\#}\left(T_{\surface}\right),\virv\right)=\Phi_{\conf}\left(\conf_{\#}\left(T_{\surface}\right),\sum_{i=1}^{n}\virv_{i}e_{i}\right)=\sum_{i=1}^{n}\Phi_{\conf}\left(\conf_{\#}\left(T_{\surface}\right),\virv_{i}e_{i}\right).
\end{equation}
Set $\Phi_{\conf}^{i}\left(\conf_{\#}\left(T_{\surface}\right),u\right)=\Phi_{\conf}\left(\conf_{\#}\left(T_{\surface}\right),ue_{i}\right)$
for all $u\in\Lip(\conf\left\{ \body\right\} )$, so that $\Phi_{\conf}^{i}$
is naturally viewed as the $i$-th component of the Cauchy flux at
the configuration $\conf$. One has, 
\begin{equation}
\Phi_{\conf}\left(\conf_{\#}\left(T_{\surface}\right),\virv\right)=\sum_{i=1}^{n}\Phi_{\conf}^{i}\left(\conf_{\#}\left(T_{\surface}\right),\virv_{i}\right).\label{eq:Cauchy_flux_components-1}
\end{equation}

\begin{description}
\item [{Balance}] There is a number $0<s<\infty$ such that for all components
of the Cauchy flux 
\begin{equation}
\ssys_{\conf}^{i}\left(\conf_{\#}\left(T_{\surface}\right),\virv\right)\leq s\norm{\virv}{\Lip,\hat{\surface}}\Fmass{\conf_{\#}\left(T_{\surface}\right)},\label{eq:Cauchy_flux_balanced}
\end{equation}
for all $\conf_{\#}\left(T_{\surface}\right)\in\conf\left(\bnd\Omega_{\body}\right)$
and $\virv\in\virvs$.
\item [{Weak~balance}] There is a number $0<b<\infty$ such that for all
components of the Cauchy flux 
\begin{equation}
\ssys_{\conf}^{i}\left(\conf_{\#}\left(\bnd T_{\part}\right),\virv\right)\leq b\norm{\virv}{\Lip,\part}\Fmass{\conf_{\#}\left(T_{\part}\right)},\label{eq:Cauchy_flux_weakly_balanced}
\end{equation}
for all $\conf_{\#}\left(T_{\part}\right)\in\conf\left(\Omega_{\body}\right)$
and $\virv\in\virvs$.
\end{description}
It is observed that from the balance property assumed above, for each
material surface $T_{\surface}$, $\ssys_{\conf}\left(\conf_{\#}\left(T_{\surface}\right),\cdot\right)$
is continuous. 
\begin{rem}
It is noted that the term $\norm{\virv}{\Lip,\hat{\surface}}$ in
the balance principle, Equation (\ref{eq:Cauchy_flux_balanced}),
may be replaced with $\norm{\virv|_{\hat{\surface}}}{\infty}=\sup_{x\in\hat{\surface}}\mass{\virv(x)}$.
We keep the former for convenience. \end{rem}
\begin{thm}
\label{thm:FluxesGiveCochains}Each component of the Cauchy flux $\Phi_{\conf}$
induces a unique flat $\left(n-1\right)$-cochain in $\conf\left\{ \body\right\} $.\label{thm:Cauchy_flux_flat_cochain}

\emph{The proof of Theorem \ref{thm:FluxesGiveCochains} will be divided
into three steps. }\end{thm}
\begin{lyxlist}{00.00.0000}
\item [{(Step~1)}] Each component of the Cauchy flux is used to defines
a linear functional on the space of polyhedral $(n-1)$-chains in
$\conf\left\{ \body\right\} $. 
\item [{(Step~2)}] The linear functional defined in the previous step
is extended to a unique flat $(n-1)$-cochain.
\item [{(Step~3)}] The compatibility of flat $(n-1)$ chains with Cauchy
flux is established.\end{lyxlist}
\begin{proof}
Step 1:

Let $\simp^{n-1}$ be an oriented $\left(n-1\right)$-simplex in $\conf\left\{ \body\right\} $.
Since $\conf\left\{ \body\right\} $ is open there exists some $n$-simplex
$\simp^{n}$ in $\conf\left\{ \body\right\} $ such that $\simp^{n-1}\subset\bnd\simp^{n}$.
Since $\conf^{-1}\left\{ \simp^{n}\right\} $ is a set of finite perimeter
in $\body$ it follows that $\simp^{n-1}\in\conf\left(\bnd\Omega_{\body}\right)$.
In other words, every oriented $(n-1)$-simplex in $\conf\left\{ \body\right\} $
may be viewed as an element of $\conf\left(\bnd\Omega_{\body}\right)$. 

In what follows, we use extensions of Lipschitz mappings as implied
by Kirszbraun's theorem. First, define a real valued function $\alpha$
of $(n-1)$-simplices. Let $u:\conf\left\{ \body\right\} \to\reals$
be a locally Lipschitz function in $\conf\left\{ \body\right\} $
such that $u(x)=1$ for $x\in\simp^{n-1}$, and we set 
\begin{equation}
\alpha\left(\simp^{n-1}\right)=\Phi_{\conf}^{i}\left(\simp^{n-1},u\right).\label{eq:Caucy_flux_on_simpices}
\end{equation}
The fact that the definition is independent of the choice of $u$
follows from condition (\ref{eq:Cauchy_flux_balanced}) and will be
demonstrated below where $\alpha$ is extended to polyhedral $(n-1)$-chains. 

Consider a polyhedral $(n-1)$-chain $\chain=\sum_{j=1}^{J}a_{j}\simp_{j}^{n-1}$
in $\conf\left\{ \body\right\} $ such that $\left\{ \simp_{j}^{n-1}\right\} _{j=1}^{J}$
are pairwise disjoint. Define the function $u:\cup_{j=1}^{J}\simp_{j}^{n-1}\to\reals$
by 
\begin{equation}
u(x)=a_{j}\quad\text{if }x\in\simp_{j}^{n-1}.
\end{equation}

We now apply Kirszbraun's theorem and obtain $\tilde{u}:\conf\left\{ \body\right\} \to\reals$,
a Lipschitz extension to $u$ defined on $\conf\left\{ \body\right\} $.
By the properties postulated for Cauchy fluxes 
\begin{eqnarray}
\Phi_{\conf}^{i}\left(\cup_{j=1}^{J}\simp_{j}^{n-1},\tilde{u}\right) & = & \sum_{j=1}^{J}\Phi_{\conf}^{i}\left(\simp_{j}^{n-1},\tilde{u}\right)=\sum_{j=1}^{J}a_{j}\alpha\left(\simp_{j}^{n-1}\right).
\end{eqnarray}
The function $\alpha$ is now extended to polyhedral $(n-1)$-chains
in $\conf\left\{ \body\right\} $ by linearity, \emph{i.e.}, 
\begin{equation}
\alpha\left(\chain\right)=\alpha\left(\sum_{j=1}^{J}a_{j}\simp_{j}^{n-1}\right)=\sum_{j=1}^{J}a_{j}\alpha\left(\simp_{j}^{n-1}\right).\label{eq:Cauchy_flux_on_polyhedral_chains}
\end{equation}
Thus, $\alpha$ is a linear functional of polyhedral $(n-1)$-chains.
The value of $\alpha(\chain)$ is independent of any particular extension
of $u$, for given $\tilde{u}',\tilde{u}$ any two Lipschitz extensions
of $u$,
\begin{multline}
\left|\Phi_{\conf}^{i}\left(\cup_{j=1}^{J}\simp_{j}^{n-1},\tilde{u}\right)-\Phi_{\conf}^{i}\left(\cup_{j=1}^{J}\simp_{j}^{n-1},\tilde{u}'\right)\right|\\
\begin{split} & =\left|\Phi_{\conf}^{i}\left(\cup_{j=1}^{J}\simp_{j}^{n-1},\tilde{u}-\tilde{u}'\right)\right|,\\
 & \leq s\norm{\tilde{u}-\tilde{u}'}{\Lip,\cup_{j=1}^{J}\simp_{j}^{n-1}}\Fmass{T_{\cup_{j=1}^{J}\simp_{j}^{n-1}}},\\
 & =0.
\end{split}
\end{multline}

Step 2:\\
From Equation (\ref{eq:Cauchy_flux_balanced}) it follows that 
\begin{equation}
\mass{\alpha\left(\simp^{n-1}\right)}\leq s\Fmass{\simp^{n-1}},\quad\text{for all}\quad\simp^{n-1}\in\conf\left\{ \body\right\} ,
\end{equation}
and by Equation (\ref{eq:Cauchy_flux_weakly_balanced}), 
\begin{equation}
\mass{\alpha\left(\bnd\simp^{n}\right)}\leq b\Fmass{\simp^{n}},\quad\text{for all}\quad\simp^{n}\in\conf\left\{ \body\right\} .
\end{equation}
The flat norm of a the functional $\alpha$ is defined by
\begin{equation}
\begin{split}F(\alpha)=\sup\left\{ \alpha(\chain)\mid\chain\,\,\text{is a polyhedarl }(n-1)\text{-chain},\right.\\
\left.\qquad\; F_{\compact}(\chain)\leq1,\:\compact\subset\conf\left\{ \body\right\} \right\} .
\end{split}
\label{eq:flat_norm_polyhedral_cochain}
\end{equation}
Using Equation (\ref{eq:Flat_norm_Whitney}) the $\compact$-flat
semi-norm of $\chain$ is given by 
\begin{equation}
F_{\compact}(\chain)=\inf_{B}\left\{ \Fmass{\chain-\bnd B}+\Fmass B\mid B-\text{a polyhedral }n\text{-chain,}\,\spt(B)\subset\compact\right\} .\label{eq:flat_norm_polehedral_chain}
\end{equation}
The flat norm of $\alpha$ is given by
\begin{equation}
F(\alpha)=\max\left\{ \sup_{\simp^{n-1}\in\conf\left\{ \body\right\} }\frac{\alpha\left(\simp^{n-1}\right)}{\Fmass{\simp^{n-1}}},\sup_{\simp^{n}\in\conf\left\{ \body\right\} }\frac{\alpha\left(\bnd\simp^{n}\right)}{\Fmass{\simp^{n}}}\right\} \leq\max\left\{ s,b\right\} .\label{eq:flat_norm_cochain_result}
\end{equation}
To obtain the last estimate, let $\epsilon>0$ and $B_{\epsilon}$
be a polyhedral $n$-chain such that $\Fmass{\chain-\bnd B_{\epsilon}}+\Fmass{B_{\epsilon}}\leq F_{\compact}(\chain)+\epsilon$,
so that 
\begin{equation}
\begin{split}\mass{\alpha\left(\chain\right)} & \leq\mass{\alpha\left(\chain-\bnd B_{\epsilon}\right)}+\mass{\alpha\left(\bnd B_{\epsilon}\right)},\\
 & \leq\Fmass{\alpha}\Fmass{\chain-\bnd B_{\epsilon}}+\Fmass{\cbnd\alpha}\Fmass{B_{\epsilon}},\\
 & \leq\sup\left\{ \Fmass{\alpha},\Fmass{\cbnd\alpha}\right\} \left(\Fmass{\chain-\bnd B_{\epsilon}}+\Fmass{B_{\epsilon}}\right),\\
 & \leq\sup\left\{ \Fmass{\alpha},\Fmass{\cbnd\alpha}\right\} \left(F_{\compact}(\chain)+\epsilon\right).
\end{split}
\end{equation}
Letting $\epsilon\to0$, it follows that 
\begin{equation}
F(\alpha)\leq\sup\left\{ \Fmass{\alpha},\Fmass{\cbnd\alpha}\right\} .
\end{equation}
By Equation (\ref{eq:F<M}) it follows that $\Fmass{\alpha}\leq\Fflat{\alpha}$
and we obtain 
\begin{equation}
F(\alpha)=\sup\left\{ \Fmass{\alpha},\Fmass{\cbnd\alpha}\right\} .
\end{equation}
Since the terms $\Fmass{\alpha},\:\Fmass{\cbnd\alpha}$ are evaluated
on polyhedral chains it is sufficient to evaluate it on simplices
and Equation (\ref{eq:flat_norm_cochain_result}) follows.

We also recall, \cite[Section 4.1.23]{Federer1969}, that polyhedral
chains form a dense subspace of the space of flat chains, specifically,
for every $\chain\in F_{n-1,\compact}\left(\reals^{n}\right)$, a
compact subset $C\subset\conf(\body)$ whose interior contain $K$
and $\varepsilon>0$, there is and a polyhedral $(n-1)$-chain $\chain_{\varepsilon}$
supported in $C$ such that 
\begin{equation}
F_{C}\left(\chain-\chain_{\varepsilon}\right)\leq\varepsilon.
\end{equation}
Thus, for every flat $(n-1)$-chain $\chain$ we have a sequence
$\chain_{j}$ such that $\lim_{i\to\infty}^{F}\chain_{j}=\chain$.
The cochain $\alpha$ is uniquely extended a flat $(n-1)$-cochain
$\cflux$ such that for every $\chain=\lim_{j\to\infty}^{F}\chain_{j}$
\begin{equation}
\cflux(\chain)=\lim_{j\to\infty}\alpha(\chain_{j}).
\end{equation}
The foregoing part of the theorem is analogous to \cite[Section V.4]{Whitney1957}.

Step 3:\\
In order to complete the proof we need to show that for $\conf_{\#}\left(T_{\surface}\right)\in\conf\left(\bnd\Omega_{\body}\right)$
and $\virv\in\SS{\conf\left\{ \body\right\} }$ we obtain $\cflux\left(\virv\conf_{\#}\left(T_{\surface}\right)\right)=\Phi_{\conf}^{i}\left(\conf_{\#}\left(T_{\surface}\right),\virv\right)$.
By \cite[Section 4.1.17]{Federer1969} the class of flat chains of
finite mass is the $M$-closure of normal currents. The chain $\virv\conf_{\#}\left(T_{\surface}\right)$
is a flat $(n-1)$-chain of finite mass. Hence, the sequence of polyhedral
$(n-1)$-chains $\left\{ \chain_{j}\right\} _{j=1}^{\infty}$, converging
$\virv\conf_{\#}\left(T_{\surface}\right)$ in the flat norm, has
a convergent sub-sequence $\left\{ \chain_{j'}\right\} _{j'=1}^{\infty}$
such that $\left\{ \chain_{j'}\right\} $ converges to $\virv\conf_{\#}\left(T_{\surface}\right)$
in the flat norm and 
\begin{equation}
\Fmass{\virv\conf_{\#}\left(T_{\surface}\right)}=\lim_{j'}\Fmass{\chain_{j'}}.
\end{equation}
 By the definition of $\alpha$ and the balance principle, Equation
(\ref{eq:Cauchy_flux_balanced}) the sequence $\left\{ \alpha\left(\chain_{j'}\right)\right\} _{j'=1}^{\infty}$
is a Cauchy sequence in $\reals$ since $\mass{\alpha\left(\chain_{m}\right)-\alpha\left(\chain_{k}\right)}\leq s\Fmass{\chain_{m}-\chain_{k}}$.
Hence 
\begin{equation}
\lim_{j'\to\infty}\alpha\left(\chain_{j'}\right)=\Phi_{\conf}^{i}\left(\conf_{\#}\left(T_{\surface}\right),\virv\right).
\end{equation}
Since $\cflux$ is an extension of $\alpha$ it follows that $\cflux(\chain_{j}')=\alpha(\chain_{j}')$
and
\begin{equation}
\begin{split}\mass{\cflux\left(\virv\conf_{\#}\left(T_{\surface}\right)\right)-\Phi_{\conf}^{i}\left(\conf_{\#}\left(T_{\surface}\right),\virv\right)} & =\mass{\cflux\left(\virv\conf_{\#}\left(T_{\surface}\right)\right)-\lim_{j'\to\infty}\alpha(\chain_{j}')},\\
 & =\mass{\cflux\left(\virv\conf_{\#}\left(T_{\surface}\right)\right)-\lim_{j'\to\infty}\cflux\left(\chain_{j'}\right)},\\
 & =\mass{\cflux\left(\virv\conf_{\#}\left(T_{\surface}\right)\right)-\cflux\left(\lim_{j'\to\infty}\chain_{j'}\right)},\\
 & =\mass{\cflux\left(\virv\conf_{\#}\left(T_{\surface}\right)-\lim_{j'\to\infty}\chain_{j'}\right)},\\
 & \leq\max\left\{ s,b\right\} \lim_{j'\to\infty}\Fflat{\virv_{i}\conf_{\#}\left(T_{\surface}\right)-\chain_{j}'}=0,
\end{split}
\end{equation}
which completes the proof.
\end{proof}
The extension of each flat $(n-1)$-cochain from $\conf\left\{ \body\right\} \subset\reals^{n}$
to $\reals^{n}$ is done trivially by setting its representing flat
$(n-1)$-form to vanish outside $\conf\left\{ \body\right\} $. We
conclude that a Cauchy flux $\ssys_{\conf}$ induces a unique $n$-tuple
of flat $(n-1)$-cochains in $\reals^{n}$ such that 
\begin{equation}
\ssys_{\conf}\left(\conf_{\#}\left(T_{\surface}\right),\virv\right)=\sum_{i=1}^{n}\cflux^{i}\left(\virv_{i}\conf_{\#}\left(T_{\surface}\right)\right),\label{eq:fluxesAndCochains}
\end{equation}
for all $\virv\in\virvs$ and $\conf_{\#}\left(T_{\surface}\right)\in\conf\left(\bnd\Omega_{\body}\right)$.
The inverse implication is provided by
\begin{thm}
\label{thm:CochainsGiveFluxes}An $n$-tuple $\{\cflux^{i}\}$ of
flat $(n-1)$ cochains in $\reals^{n}$ induces by Equation \emph{(}\ref{eq:fluxesAndCochains}\emph{)}
a unique Cauchy flux $\ssys_{\conf}$.\end{thm}
\begin{proof}
For each $\virv\in\virvs$ and $\conf_{\#}T_{\surface}$, the Cauchy
flux $\ssys_{\conf}\left(\conf_{\#}\left(T_{\surface}\right),\virv\right)$
will be defined by Equation (\ref{eq:fluxesAndCochains}), and by
the components 
\begin{equation}
\ssys_{\conf}^{i}\left(\conf_{\#}\left(T_{\surface}\right),\virv_{i}\right)=\cflux^{i}\left(\virv_{i}\conf_{\#}\left(T_{\surface}\right)\right).
\end{equation}
The additivity (\ref{eq:Cauchy_flux_additivity}) and linearity (\ref{eq:Cauchy_flux_linearity})
properties clearly hold since $\cflux^{i}$ is a linear function of
flat $(n-1)$-chains. For the Balance (\ref{eq:Cauchy_flux_balanced})
and weak balance (\ref{eq:Cauchy_flux_weakly_balanced}) properties,
recall that since $\cflux^{i}$ is a flat $(n-1)$-cochain, there
exists $C>0$ such that for every flat $(n-1)$-chain $\chain$ with
support in $\compact$, we may write $\mass{\cflux^{i}(\chain)}\leq CF_{\compact}\left(\chain\right)$.
For the balance property
\begin{equation}
\begin{split}\mass{\ssys_{\conf}^{i}\left(\conf_{\#}\left(T_{\surface}\right),\virv_{i}\right)} & =\mass{\cflux^{i}\left(\virv_{i}\conf_{\#}\left(T_{\surface}\right)\right)},\\
 & \leq CF_{\conf\left(\surface\right)}\left(\virv_{i}\conf_{\#}\left(T_{\surface}\right)\right),\\
 & \leq C\Fmass{\virv_{i}\conf_{\#}\left(T_{\surface}\right)},\\
 & \leq C\norm{\virv_{i}}{\Lip,\hat{\surface}}\Fmass{\conf_{\#}\left(T_{\surface}\right)}.
\end{split}
\end{equation}
For the weak balance
\begin{equation}
\begin{split}\mass{\ssys_{\conf}^{i}\left(\conf_{\#}\left(\bnd T_{\part}\right),\virv_{i}\right)} & =\mass{\cflux^{i}\left(\virv_{i}\conf_{\#}\left(\bnd T_{\part}\right)\right)},\\
 & \leq CF_{\conf\left(\part\right)}\left(\virv_{i}\conf_{\#}\left(\bnd T_{\part}\right)\right),\\
 & =CF_{\conf\left(\part\right)}\left(\bnd\left(\virv_{i}\conf_{\#}\left(T_{\part}\right)\right)+\cbnd\virv\irest T_{\part}\right),\\
 & \leq C\left[F_{\conf\left(\part\right)}\left(\bnd\left(\virv_{i}\conf_{\#}\left(T_{\part}\right)\right)\right)+F_{\conf\left(\part\right)}\left(\cbnd\virv\irest\conf_{\#}\left(T_{\part}\right)\right)\right],\\
 & \leq C\left[F_{\conf\left(\part\right)}\left(\virv_{i}\conf_{\#}\left(T_{\part}\right)\right)+F_{\conf\left(\part\right)}\left(\cbnd\virv\irest\conf_{\#}\left(T_{\part}\right)\right)\right],\\
 & \leq C\left[\Fmass{\virv_{i}\conf_{\#}\left(T_{\part}\right)}+\Fmass{\cbnd\virv\irest\conf_{\#}\left(T_{\part}\right)}\right],\\
 & \leq C\left[\sup_{x\in\conf(\part)}\mass{\virv_{i}(x)}\Fmass{\conf_{\#}\left(T_{\part}\right)}+n\Lip_{\virv,\conf(\part)}\Fmass{\conf_{\#}\left(T_{\part}\right)}\right],\\
 & \leq C(n+1)\norm{\virv_{i}}{\Lip,\conf(\part)}\Fmass{\conf_{\#}\left(T_{\part}\right)}.
\end{split}
\end{equation}
 
\end{proof}

Thus, Theorems \ref{thm:FluxesGiveCochains} and \ref{thm:CochainsGiveFluxes}
restate the point of view presented in \cite{Rodnay2003} that the
balance and weak-balance assumptions of stress theory may be replaced
by the requirement that the system of forces is given in terms of
an $n$-tuple of flat $(n-1)$-cochains.

%% file: generalized_bodies.tex
\chapter{Generalized bodies and Generalized surfaces\label{sec:Generalized-bodies}}

The representation of a Cauchy flux by an $n$-tuple of flat $\left(n-1\right)$-cochains
enables the generalization of the class of admissible bodies and the
introduction of a larger class of material surfaces. By a generalized
body we will mean a subset $\gpart$ of the open set $\body$ such
that the induced current $T_{\gpart}$ is a flat $n$-chain in $\body$.
Note that the general structure constructed thus far holds for generalized
bodies. For any configuration $\conf\in\Limb(\body,\reals^{n})$,
the current $\conf_{\#}\left(T_{\gpart}\right)$ is a flat $n$-chain
in $\reals^{n}$, and the operations $\cflux\left(\virv\conf_{\#}\left(\bnd T_{\gpart}\right)\right)$
and $\cbnd\cflux\left(\virv\conf_{\#}\left(T_{\gpart}\right)\right)$
are well defined.
\begin{defn}
A \textit{generalized body} is a set $\gpart\subset\body$ such that
the induced current $T_{\gpart}=\lusb^{n}\llcorner\gpart$ given by
\begin{equation}
T_{\gpart}(\omega)=\int_{\gpart}\omega d\lusb^{n},
\end{equation}
is a flat $n$-chain in $\body$. 
\end{defn}
By \cite[Section 4.1.24]{Federer1969} the current $T_{\gpart}$ is
a rectifiable $n$-current or an \textit{integral flat $n$-chain}
in $\body$. Moreover, we have 

\begin{equation}
\Fflat{T_{\gpart}}=\Fmass{T_{\gpart}}=\lusb^{n}\left(\gpart\right).
\end{equation}
It is recalled (\cite[Section 3.2.14]{Federer1969}) that a set $E$
is said to be $m$-rectifiable if there exists a Lipschitz function
mapping some bounded subset of $\reals^{m}$ onto $E$. The above
definition of generalized bodies implies that a generalized body may
be characterized as an $n$-rectifiable set in $\body$, or alternatively,
as an $\lusb^{n}$-summable set in $\body$. The class of generalized
admissible bodies is 

\begin{equation}
\gunivb=\left\{ T_{\gpart}\mid\gpart\subset\body,T_{\gpart}\in F_{n}(\body)\right\} .
\end{equation}
As mentioned in Chapter \ref{sec:Sets_of_finite_perimeter}, $\gunivb$
will have the structure of a Boolean algebra if $\body$ was postulated
to be a bounded set. Since $N_{n}(\body)\subset F_{n}(\body)$, it
is clear that $\Omega_{\body}\subset\gunivb$. Given $T_{\gpart},T_{\gpart'}\in\gunivb$
clearly $T_{\gpart\cup\gpart'}$ is an element of $\gunivb$. Contrary
to the previous definition of bodies, a generalized body needs not
be a set of finite perimeter. Although $\gpart$ is a bounded set,
its measure theoretic boundary, $\mtb(\gpart)$, may be unbounded
in the sense that $H^{n-1}(\mtb(\gpart))=\infty$. Generally speaking,
the boundary of a rectifiable set may not be a rectifiable set. A
classical example of such a generalized body in $\reals^{2}$ is the
Koch snowflake. In \cite{Silhavy2006}, such a body is referred to
as a \textit{rough body}.
\begin{rem}
It is noted that although every generalized body $\gpart$ induces
an integral flat $n$-chain, not every integral flat represents a
generalized body. However, it seems plausible that a \textit{flat
$n$-class,} introduced in \cite{Ziemer1962}, is in one to one correspondence
with the class of generalized bodies. This issue will not be considered
in this work.
\end{rem}
Considering a generalized surface, we first note that for a generalized
body $T_{\gpart}$, $\bnd T_{\gpart}$ is a flat $(n-1)$-chain in
$\body$. In addition, the following argument (\cite[Lemma 2.1]{Fleming1966})
indicates that the restrictions of flat chains to general Borel subsets
are not necessarily flat chains. Let $\halfs_{\lambda,s}$ denote
the closed half space defined by the linear functional $\lambda:\reals^{n}\to\reals$
such that 
\begin{equation}
\halfs_{\lambda,s}=\left\{ x\in\reals^{n}\mid\lambda(x)\geq s\right\} .\label{eq:half_space}
\end{equation}

For a body $T_{\gpart}$ and a closed half-space $\halfs_{\lambda,s}$
the current $T_{\gpart}\rest\halfs_{\lambda,s}$ is defined as $T_{\gpart}\rest\gamma_{\lambda,s}$
where $\gamma_{\lambda,s}$ is the characteristic function of the
half-space $\halfs_{\lambda,s}$. Since $\gamma_{\lambda,s}$ defines
a flat $0$-cochain, we may apply Equation (\ref{eq:bnd_sharp_times_chain-1})
and obtain 
\begin{equation}
\bnd\left(T_{\gpart}\rest H_{\lambda,s}\right)=\bnd T_{\gpart}\rest H_{\lambda,s}+T_{\gpart}\rest\bnd H_{\lambda,s}.
\end{equation}
 Let $T_{\gpart}\in F_{\compact,n}(\body)$ be a generalized body
in $\body$ supported in a compact subset $\compact$ of $\body$,
so that $\bnd T_{\gpart}$ is a flat $\left(n-1\right)$-chain, and
consider the chain $\bnd T_{\gpart}\rest H_{\lambda,s}$. One has,
\begin{equation}
\begin{split}F_{\compact}\left(\bnd T_{\gpart}\rest H_{\lambda,s}\right) & =F_{\compact}\left(\bnd T_{\gpart}\rest H_{\lambda,s}+\bnd\left(T_{\gpart}\rest H_{\lambda,s}\right)-\bnd\left(T_{\gpart}\rest H_{\lambda,s}\right)\right),\\
 & \leq F_{\compact}\left(\bnd T_{\gpart}\rest H_{\lambda,s}-\bnd\left(T_{\gpart}\rest H_{\lambda,s}\right)\right)+F_{\compact}\left(\bnd\left(T_{\gpart}\rest H_{\lambda,s}\right)\right),\\
 & \leq F_{\compact}\left(\bnd T_{\gpart}\rest H_{\lambda,s}-\bnd\left(T_{\gpart}\rest H_{\lambda,s}\right)\right)+F_{\compact}\left(T_{\gpart}\rest H_{\lambda,s}\right),\\
 & \leq\Fmass{\bnd T_{\gpart}\rest H_{\lambda,s}-\bnd\left(T_{\gpart}\rest H_{\lambda,s}\right)}+\Fmass{T_{\gpart}\rest H_{\lambda,s}}.
\end{split}
\end{equation}
 Since $T_{\gpart}$ is a chain of finite mass, $\Fmass{T_{\gpart}\rest H_{\lambda,s}}<\infty$.
In addition 
\begin{equation}
\int_{-\infty}^{\infty}\Fmass{\bnd T_{\gpart}\rest H_{\lambda,s}-\bnd\left(T_{\gpart}\rest H_{\lambda,s}\right)}ds=\Fmass{T_{\gpart}},
\end{equation}
and so we can show that $\Fmass{\bnd T_{\gpart}\rest H_{\lambda,s}-\bnd\left(T_{\gpart}\rest H_{\lambda,s}\right)}<\infty$
only for $\lusb^{1}$-almost every $s\in\reals$. 

In order to define a generalized material surface we follow \cite{Silhavy2006}
where the various properties of flux over fractal boundaries are investigated.
\begin{defn}
For a generalized body $\gpart$, the subset $\gsurface\subset\mtb(\gpart)$
is said to be a \textit{trace} if there exists a set of finite perimeter
$M$ such that $\gsurface=\mtb(\gpart)\cap M$ and $H^{n-1}(\mtb(\gpart)\cap\mtb\left(M\right))=0$.
Each trace $\gsurface$ is associated with a unique flat $(n-1)$-chain
$T_{\gsurface}$ given by 
\begin{equation}
T_{\gsurface}=\bnd T_{\gpart\cap M}-\bnd T_{M}\rest\gpart.\label{eq:trace_current}
\end{equation}

\end{defn}
For each $\omega\in\D^{n-1}(\body)$ we have 
\begin{equation}
T_{\gsurface}(\omega)=\int_{\gpart\cap M}d\omega(e_{1}\wedge\dots\wedge e_{n})d\lusb^{n}-\int_{\mtb\left(M\right)\cap\gpart}\omega(\vec{T}_{\bnd M})dH^{n-1},
\end{equation}
where $\vec{T}_{\bnd M}$ is defined as in Equation (\ref{eq:T_S_def_vec_def}).
The set $M$, of finite perimeter, is referred to as the generator
of the trace $\gsurface$ and it is shown in \cite{Silhavy2006} that
$\gsurface$ depends on $M$ only through the intersection of $\bnd T_{\gpart}$
with $M$. 

The collection of generalized material surfaces is defined as 
\begin{equation}
\gunivs=\left\{ T_{\gsurface}\mid\gsurface\text{ is a trace in \ensuremath{\body}}\right\} .
\end{equation}
We note that by Proposition \ref{prop:sharp_times_normal_and_flat},
for all $T_{\gsurface}\in\gunivs$ and $\virv\in\virvs$, the multiplication
$\virv\conf_{\#}\left(T_{\gsurface}\right)$ is an $n$-tuple of flat
$(n-1)$-chains. Thus, by Theorem \ref{thm:Cauchy_flux_flat_cochain}
the Cauchy flux is naturally extended to the Cartesian product $\virvs\times\conf\left(\gunivs\right)$.

%% file: Strain.tex
\chapter[The principle of virtual work]{Virtual strains and the principle of virtual work\label{sec:virtual-strains-and_virtual_work}}

For $T_{\gpart}\in\gunivs$ and $\virv\in\virvs$, $\bnd\left(\virv\conf_{\#}\left(T_{\gpart}\right)\right)$
is an $n$-tuple of flat $(n-1)$-chains in $\body$, whose components
are defined by 
\[
\left[\bnd\left(\virv\conf_{\#}\left(T_{\gpart}\right)\right)\right]_{i}=\bnd\left(\virv_{i}\conf_{\#}\left(T_{\gpart}\right)\right).
\]
Thus, $\cflux\left(\bnd\left(\virv\conf_{\#}\left(T_{\gpart}\right)\right)\right)$
is a well defined action of an $n$-tuple of flat $(n-1)$-cochains
on an $n$-tuple of flat $(n-1)$ chains. Applying Equation \ref{eq:bnd_sharp_times_chain-1}
for each component we obtain
\begin{equation}
\sum_{i=1}^{n}\cflux_{i}\left(\bnd\left(\virv_{i}\conf_{\#}\left(T_{\gpart}\right)\right)\right)=\sum_{i=1}^{n}\cflux_{i}\left(\virv_{i}\conf_{\#}\left(\bnd T_{\gpart}\right)\right)-\sum_{i=1}^{n}\cflux_{i}\left(\cbnd\alpha_{\virv_{i}}\irest\conf_{\#}\left(T_{\gpart}\right)\right).
\end{equation}
 Here $\alpha_{\virv_{i}}$ is the flat $0$-chain defined in Section
\ref{sec:Sharp_functions}.

The terms on the right-hand side of the equation above may be interpreted
as follows. The term $\sum_{i=1}^{n}\cflux_{i}\left(\virv_{i}\conf_{\#}\left(\bnd T_{\gpart}\right)\right)$
is interpreted as the virtual power performed by the surface forces
for the virtual velocity $\virv$ on the boundary of the body $T_{\gpart}$
at the configuration $\conf$. Next, for $-\cflux\left(\bnd\left(\virv\conf_{\#}\left(T_{\gpart}\right)\right)\right)=-\cbnd\cflux\left(\virv\conf_{\#}\left(T_{\gpart}\right)\right)$,
the $n$-tuple of flat $n$-cochains $-\cbnd\cflux$ is viewed as
the body force. Thus the term $-\cbnd\cflux\left(\virv\conf_{\#}\left(T_{\gpart}\right)\right)$
is interpreted as the virtual power performed by the body forces along
the virtual velocity $\virv$ on the body $T_{\gpart}$ at the configuration
$\conf$. Finally, $\sum_{i=1}^{n}\cflux_{i}\left(\cbnd\alpha_{\virv_{i}}\irest\conf_{\#}\left(T_{\gpart}\right)\right)$
is interpreted as the virtual power performed by the Cauchy flux along
the derivative of the virtual velocity $\virv$ on the body $T_{\gpart}$
at the configuration $\conf$. The last term is traditionally viewed
as the virtual power performed by the stress which we will formally
present in Chapter \ref{chap:Reynolds-transport-theorem}.

An internal virtual velocity is viewed as an element upon which the
Cauchy flux will act. Thus, a generalized \textit{internal virtual
velocity} is defined as an $n$-tuple of flat $\left(n-1\right)$-chains
in $\conf\left\{ \body\right\} $. A typical internal virtual velocity
will be denoted by $\ivirv$ and is viewed as a velocity gradient
or a linear strain-like entity. Clearly, not every internal virtual
velocity is derived from an external virtual velocity. Motivated by
the above physical interpretation and the classical formulation of
the principle of virtual work, we introduce the \textit{kinematic
interpolation map}  
\begin{equation}
\strech:\conf(\gunivb)\times\virvs\to\left[F_{n-1}\left(\conf\left(\body\right)\right)\right]^{n}
\end{equation}
such that each component is given by 
\begin{equation}
\left(\strech\left(\conf_{\#}\left(T_{\gpart}\right),\virv\right)\right)_{i}=\virv_{i}\bnd\conf_{\#}\left(T_{\gpart}\right)-\bnd\left(\virv_{i}\conf_{\#}\left(T_{\gpart}\right)\right),
\end{equation}
and for a compact set $\compact$, for which $\conf\left\{ \gpart\right\} \subset\compact$
we note that 
\begin{equation}
\begin{split}F_{\compact}\left(\virv_{i}\bnd\conf_{\#}\left(T_{\gpart}\right)-\bnd\left(\virv_{i}\conf_{\#}\left(T_{\gpart}\right)\right)\right) & \leq F_{\compact}\left(\virv_{i}\bnd\conf_{\#}\left(T_{\gpart}\right)\right)+F_{\compact}\left(\bnd\left(\virv_{i}\conf_{\#}\left(T_{\gpart}\right)\right)\right),\\
 & \leq F_{\compact}\left(\virv_{i}\bnd\conf_{\#}\left(T_{\gpart}\right)\right)+F_{\compact}\left(\virv_{i}\conf_{\#}\left(T_{\gpart}\right)\right),\\
 & \leq\left(\sup_{x\in\compact}\mass{\virv_{i}(x)}+n\Lip_{\Smap,K}\right)F_{K}\left(\bnd\conf_{\#}\left(T_{\gpart}\right)\right)\\
 & \quad\quad+\left(\sup_{x\in\compact}\mass{\virv_{i}(x)}\right)F_{K}\left(\conf_{\#}\left(T_{\gpart}\right)\right),\\
 & \leq\left(n+2\right)\norm{\virv_{i}}{\Lip,\compact}F_{\compact}\left(\conf_{\#}\left(T_{\gpart}\right)\right).
\end{split}
\label{eq:e_continuous}
\end{equation}
 Note that the map $\strech$ is disjointly additive in the first
argument and is linear in the second argument. By Equation (\ref{eq:e_continuous})
$\strech$ is continuous with respect to the flat norm of $\conf_{\#}\left(T_{\gpart}\right)$
and the $\compact$-Lipschitz semi-norm of $\virv$ for any compact
$\compact$, such that $\gpart\subset\compact$. An internal virtual
velocity $\ivirv$ is said to be \textit{compatible} if there are
$\gpart\in\gunivb$ and $\virv\in\virvs$ such that 
\begin{equation}
\ivirv=\strech\left(\conf_{\#}\left(T_{\gpart}\right),\virv\right).
\end{equation}

Given a compatible virtual internal velocity $\ivirv=\strech\left(\conf_{\#}\left(T_{\gpart}\right),\virv\right)$
we may write,
\begin{eqnarray}
\cflux\left(\strech\left(\conf_{\#}\left(T_{\gpart}\right),\virv\right)\right) & = & \sum_{i=1}^{n}\cflux_{i}\left(\virv_{i}\bnd\conf_{\#}\left(T_{\gpart}\right)-\bnd\left(\virv_{i}\conf_{\#}\left(T_{\gpart}\right)\right)\right),\nonumber \\
 & = & \sum_{i=1}^{n}\cflux_{i}\left(\virv_{i}\conf_{\#}\left(\bnd T_{\gpart}\right)\right)-\sum_{i=1}^{n}\cbnd\cflux_{i}\left(\virv_{i}\conf_{\#}\left(T_{\gpart}\right)\right),\nonumber \\
 & = & \sum_{i=1}^{n}\cbnd\alpha_{\virv_{i}}\wedge\cflux_{i}\left(\conf_{\#}\left(T_{\gpart}\right)\right),
\end{eqnarray}
and obtain 
\begin{equation}
\cflux\left(\virv\conf_{\#}\left(\bnd T_{\gpart}\right)\right)-\cbnd\cflux\left(\virv\conf_{\#}\left(T_{\gpart}\right)\right)=\cflux\left(\strech\left(\conf_{\#}\left(T_{\gpart}\right),\virv\right)\right),\label{eq:virtual_power}
\end{equation}
for all $T_{\gpart}\in\gunivb$ and $\virv\in\virvs$. We view the
last equation as a generalization of the principle of virtual power.

%% file: Stress.tex
\chapter{Stress\label{sec:Stress}}

Applying the representation theorem of flat cochains, a Cauchy flux
is represented by an $n$-tuple of flat $(n-1)$-forms in $\conf\left\{ \body\right\} $.
Let $\cflux_{i}$ denote the flat $(n-1)$-cochain associated with
the $i$-th component of the Cauchy flux. Then, $\fform{\cflux_{i}}$
will be used to denote its representing flat $(n-1)$-form. The $n$-tuple
of flat $(n-1)$-forms in $\conf\left\{ \body\right\} $ representing
the Cauchy flux will be denoted by $\fform{\cflux}$ and will be referred
to as the \textit{Cauchy stress}.

Using the representation theorem for flat forms we obtain an integral
representation of the principle of virtual power given in Equation
(\ref{eq:virtual_power}). The virtual power performed by surface
forces is represented by
\begin{multline}
\sum_{i=1}^{n}\cflux_{i}\left(\virv_{i}\conf_{\#}\left(\bnd T_{\gpart}\right)\right)\\
\begin{split} & =\sum_{i=1}^{n}\left(\conf^{\#}\left(\cbnd\left(\alpha_{\virv_{i}}\wedge\cflux_{i}\right)\right)\right)\left(T_{\gpart}\right),\\
 & =\sum_{i=1}^{n}\int_{\gpart}\wcbd\left(\virv_{i}D_{\cflux_{i}}\left(\conf\left(x\right)\right)\right)\left(D\conf(x)(e_{1})\wedge\dots\wedge D\conf(x)(e_{n})\right)d\lusb_{x}^{n},\\
 & =\sum_{i=1}^{n}\int_{\gpart}\wcbd\left(\virv_{i}D_{\cflux_{i}}\left(\conf\left(x\right)\right)\right)\left(e_{1}\wedge\dots\wedge e_{n}\right)J_{\conf}(x)d\lusb_{x}^{n}.
\end{split}
\label{eq:material_integral_surface_power}
\end{multline}
Equations (\ref{eq:pullback_flat_cochain}) and (\ref{eq:pullback_flat_form})
were used in the first and second lines. As noted above, the Cauchy
stress $\fform{\cflux}$ is an $n$-tuple of flat $(n-1)$-forms,
and by the definition of flat forms (Definition \ref{def:flat_forms}),
each component of the stress is an essentially bounded, $\lusb^{n}$-integrable,
$(n-1)$-form whose weak exterior derivative is an essentially bounded,
$\lusb^{n}$-integrable $n$-form. By applying Equation (\ref{eq:bnd_sharp_times_chain-1})
to the integrand of Equation \ref{eq:material_integral_surface_power},
we note that 
\begin{equation}
\wcbd\left(\virv_{i}D_{\cflux_{i}}\right)=\wcbd\virv_{i}\wedge D_{\cflux_{i}}+\virv_{i}\wedge\wcbd D_{\cflux_{i}}.
\end{equation}
Thus, $D_{\cflux_{i}}$ and $\wcbd D_{\cflux_{i}}$ may be changed
on a set of $\lusb^{n}$-measure zero without affecting the value
of the Cauchy flux on $\conf_{\#}\left(\bnd T_{\gpart}\right)$ for
any virtual velocity $\virv$. For a generalized material surface
$T_{\gsurface}$, the representation of the stress as an equivalence
class of $\lusb^{n}$-integrable functions, may seem to be problematic
as the current $T_{\gsurface}$ is supported on a set of $L^{n}$-measure
zero. It may appear as though one can change the Cauchy flux without
changing its representing flat from. In order to resolve this issue,
we note that in order to apply the integral representation of the
Cauchy flux we must first apply Theorem \ref{thm:Federer_representation_flat_chains}
and represent the chain $T_{\gsurface}$ by Lebesgue integrable vector
fields in the form $T_{\gsurface}=\lusb^{n}\wedge\eta+\bnd\left(\lusb^{n}\wedge\xi\right).$
Thus, changing the flat form of a set of $\lusb^{n}$-measure zero
will not effect the Cauchy flux.

The virtual power performed by body forces is represented by
\begin{multline}
-\sum_{i=1}^{n}\cbnd\cflux_{i}\left(\virv_{i}\conf_{\#}\left(T_{\gpart}\right)\right)\\
\begin{split} & =-\sum_{i=1}^{n}\conf^{\#}\left(\alpha_{\virv_{i}}\wedge\cbnd\cflux_{i}\right)\left(T_{\gpart}\right),\\
 & =-\sum_{i=1}^{n}\int_{\gpart}\left(\virv_{i}\wcbd D_{\cflux_{i}}\left(\conf\left(x\right)\right)\right)\left(D\conf(x)(e_{1})\wedge\dots\wedge D\conf(x)(e_{n})\right)d\lusb_{x}^{n},\\
 & =-\sum_{i=1}^{n}\int_{\gpart}\left(\virv_{i}\wcbd D_{\cflux_{i}}\left(\conf\left(x\right)\right)\right)\left(e_{1}\wedge\dots\wedge e_{n}\right)J_{\conf}(x)d\lusb_{x}^{n}.
\end{split}
\label{eq:material_integrable_body_power}
\end{multline}

The virtual power performed by internal forces is represented by
\begin{multline}
\sum_{i=1}^{n}\cbnd\virv_{i}\wedge\cflux_{i}\left(\conf_{\#}\left(T_{\gpart}\right)\right)\\
\begin{split} & =\sum_{i=1}^{n}\conf^{\#}\left(\cbnd\alpha_{\virv_{i}}\wedge\cflux_{i}\right)\left(T_{\gpart}\right),\\
 & =\sum_{i=1}^{n}\int_{\gpart}\left(\wcbd\virv_{i}\wedge D_{\cflux_{i}}\left(\conf\left(x\right)\right)\right)\left(D\conf(x)(e_{1})\wedge\dots\wedge D\conf(x)(e_{n})\right)d\lusb_{x}^{n},\\
 & =\sum_{i=1}^{n}\int_{\gpart}\left(\wcbd\virv_{i}\wedge D_{\cflux_{i}}\left(\conf\left(x\right)\right)\right)\left(e_{1}\wedge\dots\wedge e_{n}\right)J_{\conf}(x)d\lusb_{x}^{n}.
\end{split}
\label{eq:material_integrable_stress_power}
\end{multline}

For $\conf:\body\to\reals^{n}$, a Lipschitz map, $\conf^{\#}\cflux$
is an $n$-tuple of flat $(n-1)$-cochains in $\body$. Each cochain
$\conf^{\#}\cflux_{i}$ is represented by a flat $(n-1)$-form $\fform{\conf^{\#}\cflux_{i}}=\conf^{\#}\fform{\cflux_{i}}$.
The associated $n$-tuple of flat $(n-1)$-forms, $\conf^{\#}\fform{\cflux}$
is identified as the \textit{Piola-Kirchhoff stress}
\begin{equation}
\left(\conf^{\#}\fform{\cflux}(x)\right)_{i}=J_{\conf}(x)D_{\cflux_{i}}\left(\conf(x)\right).
\end{equation}